\documentclass{egpubl}
\usepackage{eg2023}
 
\ConferencePaper        
\CGFStandardLicense

\usepackage[T1]{fontenc}
\usepackage{dfadobe}  

\usepackage{cite}  
\BibtexOrBiblatex
\electronicVersion
\PrintedOrElectronic
\ifpdf \usepackage[pdftex]{graphicx} \pdfcompresslevel=9
\else \usepackage[dvips]{graphicx} \fi

\usepackage{egweblnk} 

\usepackage{bbold}

\usepackage{amsfonts}
\usepackage{amsmath}
\usepackage{amssymb}
\usepackage{overpic}
\usepackage{pgf,tikz}
\usepackage{bm}
\usepackage{wrapfig}
\usepackage{multirow}
\usepackage{mathtools}
\usepackage{enumerate}
\usepackage{verbatim}

\newtheorem{lemma}{Lemma}[section]
\newtheorem{proposition}{Proposition}[section]
\newtheorem{definition}{Definition}[section]

\usepackage{algorithm}
\usepackage{algpseudocode}

\usepackage{xcolor}
\definecolor{tabblue}{rgb}{0.54, 0.81, 0.94}
\usepackage{listings}
\usepackage{color, colortbl}
\definecolor{mygreen}{RGB}{28,172,0} 
\definecolor{mylilas}{RGB}{170,55,241}
\definecolor{myblue}{RGB}{203,192,211}
\definecolor{mycolor1}{RGB}{255,204,201}

\usepackage{booktabs}       
\newlength{\Oldarrayrulewidth}

\newcommand{\uargmin}[1]{\underset{#1}{\text{argmin}}\;}
\newcommand{\uargmax}[1]{\underset{#1}{\text{argmax}}\;}

\def\*#1{\mathbf{#1}}

\newcommand{\inpar}[1]{\left(#1\right)}
\newcommand{\inparsmall}[1]{\big(#1\big)}

\newcommand{\eps}{\varepsilon}
\newcommand{\al}{\alpha}

\newcommand{\RR}{\mathbb{R}}

\newcommand{\Ff}{\mathcal{F}}

\newcommand{\Ii}{\mathcal{I}}

\newcommand{\Mm}{\mathcal{M}}
\newcommand{\Nn}{\mathcal{N}}

\newcommand{\Ss}{\mathcal{S}}

\newcommand{\lb}{\lambda}

\definecolor{C1}{HTML}{EC3008}

\title{Scalable and Efficient Functional Map\\ Computations on Dense Meshes}
\author[Robin Magnet \& Maks Ovsjanikov]
{\parbox{\textwidth}{\centering Robin Magnet$^{1}$\orcid{0000-0002-2192-411X}
        and Maks Ovsjanikov$^{1}$\orcid{0000-0002-5867-4046}
}
        \\
{\parbox{\textwidth}{\centering $^1$LIX, Ecole Polytechnique, IP Paris
       }
}
}

\begin{document}


\maketitle
\begin{abstract}
   We propose a new scalable version of the functional map pipeline that allows to efficiently compute correspondences between potentially very dense meshes. Unlike existing approaches that process  dense meshes by relying on ad-hoc mesh simplification, we establish an integrated end-to-end pipeline with theoretical approximation analysis. In particular, our method overcomes the computational burden of both computing the basis, as well the functional  and pointwise correspondence computation by approximating the functional spaces and the functional map itself. Errors in the approximations are controlled by theoretical upper bounds assessing the range of applicability of our pipeline. With this construction in hand, we propose a scalable practical algorithm and demonstrate results on dense meshes, which approximate those obtained by standard functional map algorithms at the fraction of the computation time. Moreover, our approach outperforms the standard acceleration procedures by a large margin, leading to accurate results even in challenging cases.


\end{abstract}  


\section{Introduction}
\label{sec:intro}
Processing and analyzing complex 3D objects is a major area of study with applications in computer graphics, medical imaging and other domains. The underlying structure of such data can be highly detailed and require dense point sets and meshes to capture important features. At the same time, shape analysis methods are often designed to only handle objects that consist of tens of thousands of points, thus requiring decimation algorithms to process meshes containing millions of points that can arise in real-world applications. While mesh simplification can lead to good results, it suffers from several drawbacks. First, the simplification process might lead to artifacts and significant loss of detail. Second, for many applications, it remains highly non-trivial to accurately transfer the results of analysis from the simplified to original shapes. Finally, the transfer process can introduce errors and aliasing artifacts.

\begin{figure}
  \centering
  \includegraphics[width=.9\linewidth]{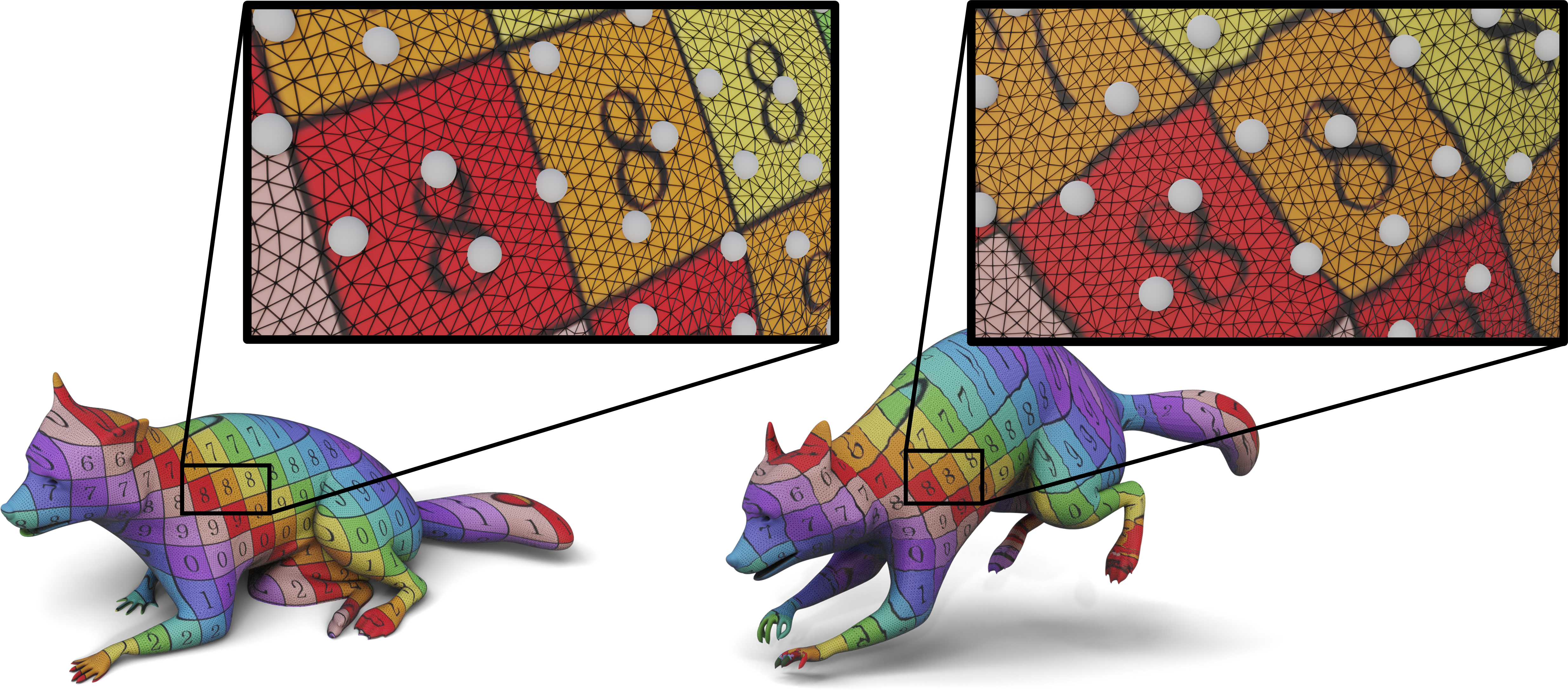}
  \caption{\label{fig:teaser}
           Our method produces point-to-point correspondences between dense meshes \textit{efficiently}, using values only located at sparse samples, displayed in white. The source and target shapes from the \textsc{DeformingThings4D} dataset~\cite{li20214dcomplete} are composed of roughly $100\ 000$ vertices, and correspondences are displayed using texture transfer. The map computation (including all preprocessing) took 60 seconds on a standard machine.}
\end{figure}
In this work, we focus on computing correspondences between non-rigid shapes. This is a long-standing problem in Geometry Processing and related fields, with a wide range of techniques developed in the past few years \cite{deng2022survey,sahilliogluRecentAdvancesShape2020}.
A notable line of work in this domain uses the so-called functional map framework, which is based on manipulating correspondences as matrices in a reduced basis \cite{ovsjanikovFunctionalMapsFlexible2012}. Methods based on this framework have recently achieved high accuracy on a range of difficult non-rigid shape matching tasks \cite{melzi2019shrec,melziZoomOutSpectralUpsampling2019,dyke2020shrec}. Unfortunately, these approaches require costly and time-consuming precomputation of the Laplacian basis and, potentially, other auxiliary data-structures \cite{renContinuousOrientationpreservingCorrespondences2019}. As a result, these techniques do not scale well to densely sampled meshes and, thus, are most often applied on simplified shapes. Moreover, while accelerated versions of some methods \cite{melziZoomOutSpectralUpsampling2019} have recently been proposed, these lack theoretical approximation guarantees, and can be error-prone.

At the same time, several approaches have recently been proposed for efficient approximation of the Laplace-Beltrami basis \cite{nasikunFastApproximationLaplaceBeltrami2018,nasikun2022hierarchical}. These approaches can successfully scale to very large meshes, and are especially effective for computing low frequency eigenfunctions. While these methods have been shown to be efficient when, e.g., using approximated spectra as shape descriptors \cite{reuter2006laplace} or for individual shape processing, they can come short when applied in \textit{shape correspondence scenarios}. Conceptually, this is because the objectives and guarantees in \cite{nasikunFastApproximationLaplaceBeltrami2018,nasikun2022hierarchical} only apply at a \textit{global} scale of individual shapes, instead of the local function approximation or function transfer required for functional and point-to-point map computation.

In this work, we make a step towards creating scalable and efficient non-rigid shape correspondence methods, which can handle very large meshes, and are backed by theoretical approximation bounds. We focus on the functional map framework \cite{ovsjanikovComputingProcessingCorrespondences2017} and especially its recent variants based on spectral upsampling, such as the ZoomOut method \cite{melziZoomOutSpectralUpsampling2019} and its follow-up works \cite{huangConsistentZoomOutEfficient2020,xiangDualIterativeRefinement2021,renDiscreteOptimizationShape2021}. These methods are based on iteratively updating functional and point-to-point maps and have been shown to lead to high-quality results in a wide range of cases. Unfortunately, the two major steps: basis pre-computation and interative updating of the pointwise maps can be costly for dense shapes. 

To address this challenge, we propose an integrated pipeline that helps to make both of these steps scalable and moreover comes with approximation guarantees. 
For this we first establish a new functional space inspired by~\cite{nasikunFastApproximationLaplaceBeltrami2018}, and demonstrate how it can be used to define an approximation of functional maps without requiring either a dense pointwise correspondence or a even basis on the dense meshes. We then provide theoretical approximation bounds for this construction that, unlike the original definition in \cite{ovsjanikovFunctionalMapsFlexible2012} is fully agnostic to the number of points in the original mesh.
Following this analysis, we extend the approach introduced in \cite{nasikunFastApproximationLaplaceBeltrami2018} to improve our functional map approximation, and present an efficient and scalable algorithm for map refinement, based on our constructions, which eventually produces accurate results in the fraction of the time required for standard processing, as displayed on Figure~\ref{fig:pipeline}.

\begin{figure*}[tbp]
  \centering
  \mbox{} \hfill
  \includegraphics[width=.8\linewidth]{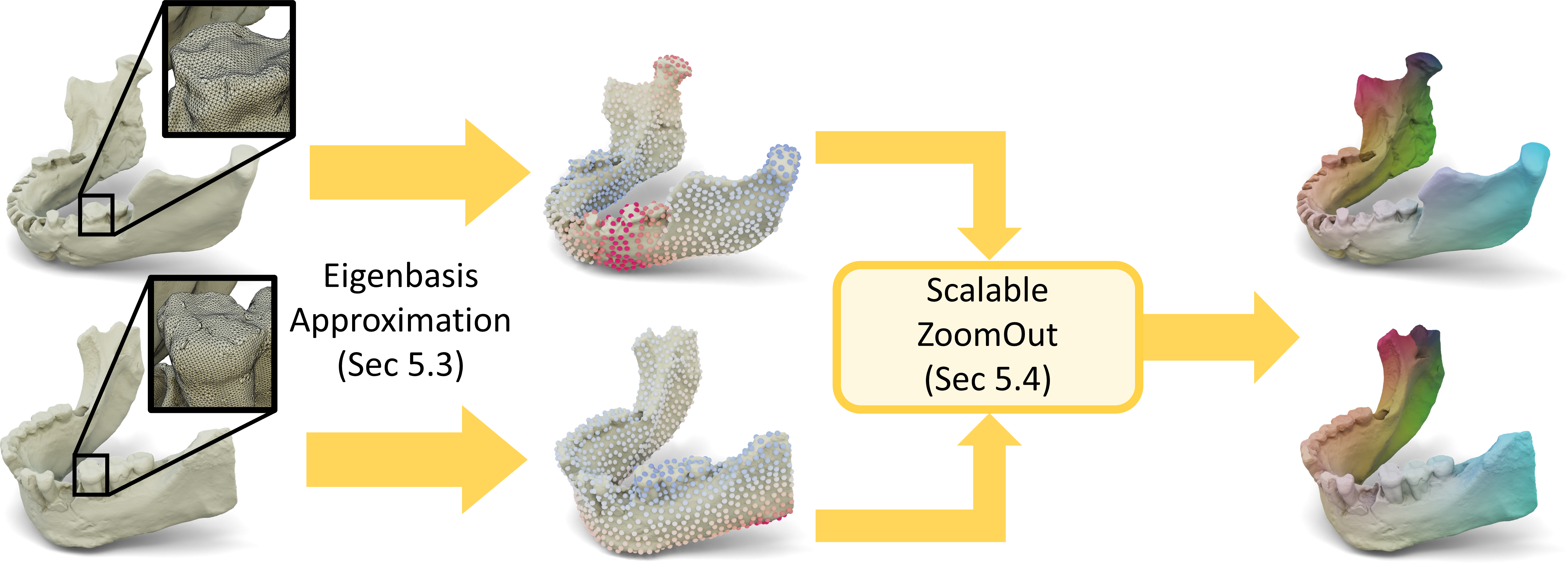}
  \hfill \mbox{}
  \caption{\label{fig:pipeline}%
  Overall pipeline of our method, using real data from~\cite{renGeometricAnalysisShape2020}. Given two dense input shapes, we first generate an approximate eigenbasis computation by using a modified version of the approach introduced in \cite{nasikun2022hierarchical} (Sec.~\ref{sec:new eigenbasis}). We then propose a new scalable version of ZoomOut (Sec.~\ref{sec:final algorithm}), which exploits our functional map approximation (Sec~\ref{sec:new FM}) and comes with theoretical appromation bounds. Ultimately, this leads to dense pointwise correspondences between the two input shapes visualized here via color transfer.
  }
\end{figure*}
\section{Related Works}
\label{sec:related works}
Our main focus is on designing a scalable and principled approach for non-rigid shape correspondence, within the functional map framework. We therefore review works that are most closely related to ours, especially those using spectral techniques for shape matching, and refer the interested readers to recent surveys \cite{van2011survey,tam2012registration,biasotti2016recent,sahilliogluRecentAdvancesShape2020,deng2022survey} for a more comprehensive overview of other approaches.

\paragraph*{Spectral methods in shape matching}
The idea of using the spectral properties of the Laplace-Beltrami, and especially its eigenfunctions for shape correspondence has been investigated in many existing works. Early approaches focused on directly aligning the eigenfunctions, seen as descriptors, \cite{mateus2008articulated,jain2007non} or using diffusion processes to derive descriptors or embedding spaces, e.g., \cite{sharma2010shape,ovsjanikov2010}, among others.

A more principled framework was introduced in
\cite{ovsjanikovFunctionalMapsFlexible2012}, based on the idea of functional maps. The overall strategy is to express the pull-back of functions as an operator in a reduced basis, and to formulate objective functions based on desirable properties of such an operator. The main advantage of this approach is that it leads to small-scale optimization problems, with the number of unknowns independent of the size of the underlying meshes.

Despite the simplicity of the original approach, its performance is strongly dependent on accurate descriptors and hyper-parameter tuning. As a result, this basic strategy has been extended significantly in many follow-up works,  based both on geometric insights \cite{kovnatskyCoupledQuasiharmonicBases2012,aflaloSpectralMultidimensionalScaling2013,ovsjanikovShapeMatchingQuotient2013,burghard2017embedding,eynard2016coupled}, improved optimization strategies \cite{kovnatsky2016madmm,nognengInformativeDescriptorPreservation2017,renMapTreeRecoveringMultiple2020,renDiscreteOptimizationShape2021}, and richer correspondence models going beyond isometries across complete shapes, \cite{rodolaPartialFunctionalCorrespondence2017,rustamovMapbasedExplorationIntrinsic2013,litanyFullySpectralPartial2017}, among others. 

\textbf{Functional and pointwise maps} While many approaches in the functional map literature focus on the optimization in the spectral domain, it has also been observed that the \emph{interaction} between pointwise and functional correspondences can lead to significant improvement in practice. This was used in the form of the Iterative Closest Point (ICP) refinement in the original article and follow-up works \cite{ovsjanikovFunctionalMapsFlexible2012,maronPointRegistrationEfficient2016,ovsjanikovShapeMatchingQuotient2013} and has since then been extended to map deblurring and denoising \cite{ezuzDeblurringDenoisingMaps2017}, as well as powerful refinement, and even map optimization strategies \cite{melziZoomOutSpectralUpsampling2019,renContinuousOrientationpreservingCorrespondences2019,huangConsistentZoomOutEfficient2020,eisenbergerSmoothShellsMultiScale2020,renDiscreteOptimizationShape2021}. All of these works are based on the insight that manipulating maps in \textit{both} the spectral and spatial (primal) domains can lead to overall improvement in the quality of the results.

Unfortunately, such approaches can often come at a cost of scalability, since the complexity of pointwise maps is directly dependent on the mesh resolution, making it difficult to scale them to highly dense meshes.

\paragraph*{Multi-resolution spectral approaches}
Our work is also related to multi-resolution techniques for approximating spectral quantities, as, e.g., in \cite{vaxman2010multi}, and especially to recent developments for accurate and scalable eigen-solvers geared towards Laplacian eigenfunctions on complex meshes \cite{nasikunFastApproximationLaplaceBeltrami2018,nasikun2022hierarchical}. The latter set of methods have been shown to lead to excellent performance and scalability on tasks involving individual shapes, such as computing their Shape-DNA \cite{reuter2006laplace} descriptors, or performing mesh filtering. Similarly, there exist several spectral coarsening and simplification approaches \cite{liuSpectralCoarseningGeometric2019,lescoatSpectralMeshSimplification2020,honglin2020chordal}  that explicitly aim to coarsen operators, such as the Laplacian while preserving their low frequency eigenapairs. Unfortunately, these methods typically rely on the eigenfunctions on the dense shapes, while the utility of the former approaches in the context of \textit{functional maps} has not yet been fully analyzed and exploited, in part, since, as we show below, this requires \textit{local} approximation bounds.

Finally, we  mention that our work is also related to hierarchical techniques, including functional maps between subdivision surfaces proposed in \cite{shoham2019hierarchical}, and even more closely, to refinement via spectral upsampling \cite{melziZoomOutSpectralUpsampling2019}. However, the former approach relies on a subdivision hierarchy, while the acceleration strategy of the latter, as we discuss below, is based on a scheme that unfortunately can fail to converge in the in the presence of full information.

\paragraph*{Limitations of existing techniques and our contributions}
To summarize, the scalability of existing functional maps-based methods is typically limited by two factors: first, the pre-processing costs associated with the computation of the eigenfunctions of the Laplace-Beltrami operator, and second, the complexity of simultaneously manipulating pointwise and functional correspondences. 

In this context, our key contributions include:
\begin{enumerate}
    \item We define an approximation of the functional map, which requires only a sparse correspondence, and provide a theoretical basis for this construction.
    \item We analyze the basis approximation approach in~\cite{nasikunFastApproximationLaplaceBeltrami2018} for functional map computation, obtaining explicit theoretical upper bounds. We then modify this approach to improve the approximation guarantees, leading to more accurate maps.
    \item We present a principled and scalable algorithm for functional map refinement, based on our constructions, which produces accurate results at a fraction of the time of comparable methods.
\end{enumerate}

\section{Method Overview}
\label{sec:motivation and overview}
As mentioned above, our overall goal is to design a scalable pipeline for non-rigid shape matching that can handle potentially very dense meshes. We base our approach on the ZoomOut variant of the functional map famework \cite{melziZoomOutSpectralUpsampling2019}. However, our constructions can be easily extended to other recent functional maps methods, e.g., \cite{renMapTreeRecoveringMultiple2020,renDiscreteOptimizationShape2021}, which share the same general algorithmic structure. Specifically, ZoomOut and related methods are based on two main building blocks: computing the eigenfunctions of the Laplace-Beltrami operator first, and then iterating between updating the point-to-point and functional maps.

Our general pipeline is displayed on Figure~\ref{fig:pipeline} and consists of the following major steps. First, we generate for each shape a sparse set of samples and a factorized functional space using a modification of the approach introduced in~\cite{nasikunFastApproximationLaplaceBeltrami2018}, described in Sec~\ref{sec:new eigenbasis}. Secondly, we use the approximation of the functional map that we introduce (Sec. \ref{sec:new FM}) to define a scalable version of the ZoomOut algorithm producing a sparse pointwise map. Finally, we extend this sparse map to a dense pointwise map with sub-sample accuracy, by using the properties of the functional subspaces we consider.

%
%

The rest of the paper is organized as follows: in Section~\ref{sec:background} we introduce the notations and background necessary for our approach. 

In Section~\ref{sec:new FM}, we introduce our functional map approximation based on the basis construction approach in ~\cite{nasikunFastApproximationLaplaceBeltrami2018}. Section \ref{sec:approximation bounds} provides explicit approximation errors and Section~\ref{sec:new eigenbasis} describes our modification of the method of~\cite{nasikunFastApproximationLaplaceBeltrami2018}, which helps to improve the theoretical upper bounds we obtained for functional map computation.
%
Given these constructions, we show in Sec.~\ref{sec:final algorithm} how ZoomOut-like algorithms can be defined, first by iteratively updating functional and pointwise maps in the reduced functional spaces, and then how the computed functional map can be extended onto the dense shapes efficiently.  

Section~\ref{sec:implementation} provides implementation details, while Section~\ref{sec:results} is dedicated to extensive experimental evaluation of our approach.

\section{Notations \& Background}
\label{sec:background}
\subsection{Notations}

For a triangle mesh, we denote by $\*W$ and $\*A$ its stiffness and mass matrices that together define the (positive semi-definite) Laplace Beltrami Operator as $L = \*A^{-1} \*W$ . Given two shapes $\Nn$ and $\Mm$ with, respectively, $n$ and $m$ vertices, any vertex-to-vertex map $T:\Nn\to\Mm$ can be represented as a binary matrix $\bm{\Pi}\in\{0,1\}^{n\times m}$ with $\bm{\Pi}_{ij}=1$ if and only if $T(x_i) = y_j$, where $x_i$ denotes the $i$-th vertex on $\Nn$ and $y_j$ the $j$-th vertex in $\Mm$.
%

The eigenfunctions of the Laplace Beltrami operator can be obtained by solving a generalized eigenproblem:
\begin{equation}\label{eq:standard LBO decomposition}
	\*W\psi_i = \lb_i \*A \psi_i,
\end{equation}
where in practice, we typically consider the eigenfunctions corresponding to the $K$ smallest eigenvalues.

\subsection{Functional Maps and the ZoomOut algorithm}

Functional maps were introduced in~\cite{ovsjanikovFunctionalMapsFlexible2012} as a means to perform dense non-rigid shape matching. 
The key insight is that any pointwise map $T:\Nn\to\Mm$ can be transformed into a functional map via composition $F_T:f\in\Ff(\Mm)\mapsto f\circ T\in\Ff(\Nn)$, where $\Ff(\Ss)$ is the space of real-valued functions on a surface $\Ss$. Since $F_T$ is linear, it can be represented as a matrix in the given basis for each space $\big(\psi_i^\Mm\big)_i$ and $\big(\psi_i^\Nn\big)_i$.

If the basis on shape $\Nn$ is orthonormal with respect to $\*A^\Nn$, the functional map $\*C$ can be expressed in the truncated basis of size $K$ on each shape as a $K\times K$ matrix:
\begin{equation}
	\label{eq:fmap_definition}
  \*C = \big(\bm{\Psi}^\Nn\big)^\top \*A^\Nn \bm{\Pi} \bm{\Psi}^\Mm,
\end{equation}
where each basis function on $\Mm$ (resp. $\Nn$) is stacked as a column of $\bm{\Psi}^\Mm$ (resp. $\bm{\Psi}^\Nn$), $\bm{\Pi}$ is the matrix representing the underlying pointwise map, and we use $^\top$ to denote the matrix transpose.

\paragraph*{ZoomOut} Given the Laplace-Beltrami eigenbasis, the ZoomOut algorithm~\cite{melziZoomOutSpectralUpsampling2019} allows to recover high-quality correspondences starting from an approximate initialization, by iterating between two steps: (1) Converting a $k \times k$ functional map to a pointwise map, (2) converting the pointwise map to a functional map of size $k+1 \times k+1$. This method has also been extended to other settings, to both promote cycle consistency \cite{huangConsistentZoomOutEfficient2020} and optimize various energies \cite{renDiscreteOptimizationShape2021} among others. Unfortunately, although simple and efficient, the scalability of this approach is limited, first, by the precomputation of the Laplacian basis, and second by the pointwise map recovery which relies on possibly expensive nearest-neighbor search queries across dense meshes.

Several ad-hoc acceleration strategies have been proposed in ~\cite{melziZoomOutSpectralUpsampling2019}. However, as we discuss below, these do not come with approximation guarantees and indeed can fail to converge in the limit of complete information.

\subsection{Eigenbasis approximation} 
\label{sec:background:approx}

To improve the scalability of spectral methods, recent works~\cite{nasikunFastApproximationLaplaceBeltrami2018, xuFastCalculationLaplaceBeltrami2021} have tried to develop approximations of the Laplace Beltrami eigenbasis, via the reduction of the search space.
Specifically, in~\cite{nasikunFastApproximationLaplaceBeltrami2018}, the authors first sample a set of $p$ points $\Ss = \{v_1, \dots, v_p\}$ on shape $\Mm$ and create a set of $p$ local functions $\inpar{u_1,\dots,u_p}$, each centered on a particular sample point.
Each function $u_j$ is built from an unnormalized function $\Tilde{u}_j$ supported on a geodesic ball of radius $\rho$ around the sample $v_j$, which decreases with the geodesic distance from the center:
\begin{equation}\label{eq:tilde u def}
    \Tilde{u}_j : x\in\Mm \mapsto \chi_\rho\inpar{d^\Mm(x,v_j)}\in\RR
\end{equation}
where $d^\Mm$ is the geodesic distance on shape $\Mm$ and $\chi_\rho:\RR_+\to\RR$ is a differentiable non-increasing function with $\chi_\rho(0)=1$ and $\chi_\rho(x)=0$ for $x\geq \rho$. Choices for $\chi$ are discussed in Appendix~\ref{app:chi function}. 
Finally, local functions $u_j$ are defined to satisfy the partition of the unity by using:
\begin{equation}\label{eq:normalization of u}
    u_j(x) = \frac{\Tilde{u}_j(x)}{\sum_k \Tilde{u}_k(x)} ~\forall~x\in\Mm
\end{equation}

Now only considering functions that lie in the $\text{Span}\left\{u_1,\dots, u_p\right\}$, the original eigendecomposition system in Eq.~\eqref{eq:standard LBO decomposition} reduces to a generalized eigenproblem of size $p\times p$:
\begin{equation}\label{eq:reduced LBO decomposition}
    \overline{\*W}\ \overline{\phi}_i = \bar{\lb}_i \overline{\*A}\  \overline{\phi}_i
\end{equation}
with $\overline{\*W}=\*U^\top\*W \*U$ and $\overline{\*A}=\*U^\top\*A \*U$ where $\*W$ and $\*A$ are the stiffness and area matrices of $\Mm$, and $\*U$ a \textit{sparse} matrix whose columns are values of functions $\{u_j\}_j$. Eigenvectors $\overline{\phi}_i$ are $p$-dimensional vectors describing the coefficients with respect to $\{u_j\}$, which define the approximated eigenvectors as $\overline{\psi}_i = \*U\overline{\phi}_i$. Note that since $\overline{\phi}_i$ are orthonormal with respect to $\overline{\*A}$ this implies that $\overline{\psi}_i$ are orthonormal with respect to $\*A$.

While the original work \cite{nasikunFastApproximationLaplaceBeltrami2018} focused on global per-shape applications such as filtering and Shape-DNA \cite{reuter2006laplace} computation, we build on and modify this pipeline in order to obtain reliable functions to perform dense \textit{shape correspondence}.

\section{Our approach}
\label{sec:contributions}
    In this section, we first present a functional map definition using the basis approximation strategy from of~\cite{nasikunFastApproximationLaplaceBeltrami2018}, and provide theoretical bounds on the approximation error (Secs.~\ref{sec:new FM} and \ref{sec:approximation bounds} respectively). Based by these results, we then introduce our modification of~\cite{nasikunFastApproximationLaplaceBeltrami2018} in Section~\ref{sec:new eigenbasis} which we use in our approach in order to minimize the computed bound. Finally we present our Extended ZoomOut algorithm and provide implementation details in Sections~\ref{sec:final algorithm} and \ref{sec:implementation}.
    
    \subsection{Approximate Functional Map}
    \label{sec:new FM}

As mentioned in Sec. \ref{sec:background:approx} the eigenfunctions computed using the approach in \cite{nasikunFastApproximationLaplaceBeltrami2018} are, by construction, orthonormal with respect to the area matrix $\*A$. Thus, they can be used to compute a functional map following Eq.~\eqref{eq:fmap_definition}. This leads to the following definition:

\begin{definition}\label{def:Functional map on new space}
Given two shapes $\Mm$ and $\Nn$ with approximated eigenfunctions $\big(\Psi_i^\Mm\big)_i$ stacked as columns of matrix $\overline{\bm{\Psi}}^\Mm$ (resp. with $\Nn$), the \textit{reduced} functional map 
associated to a pointwise map $\bm{\Pi}:\Nn\to\Mm$ is defined as:
\begin{equation}\label{eq:Approximated FM}
    \overline{\*C} = \inparsmall{\overline{\bm{\Psi}}^\Nn}^\top \*A^\Nn \bm{\Pi} \overline{\bm{\Psi}}^\Mm
\end{equation}
\end{definition}

Note that this functional map definition uses the approximated bases. However, it still relies on the knowledge of a full point-to-point map between complete (possibly very dense) shapes.
%
%
To alleviate this constraint, we introduce another functional map $\overline{\*C}$ that only relies on maps between samples, independently from the original number of points:
\begin{definition}\label{def:reduced FM}
Using the same setting as in Definition~\ref{def:Functional map on new space}, with eigenfunctions arising from Eq.~\eqref{eq:reduced LBO decomposition}, $\big(\overline{\phi}_i^\Mm\big)_i$ (resp. with $\Nn$) being stacked as columns of a matrix $\overline{\bm{\Phi}}^\Mm$ (resp. with $\Nn$), 
given a point-wise map $\overline{\bm{\Pi}}:\Ss^\Nn\to\Ss^\Mm$, our \textit{restricted} 
functional map 
is defined as:
\begin{equation}\label{eq:reduced FM}
    \widehat{\*C} = \inpar{\overline{\bm{\Phi}}^\Nn}^\top \overline{\*A}^\Nn \overline{\bm{\Pi}}\ \overline{\bm{\Phi}}^\Mm
\end{equation}
\end{definition}

Recall that, as mentioned in Sec~\ref{sec:background:approx} $\Ss$ denotes the sparse set of samples on each shape. Therefore, in order to define $\widehat{\*C}$, we only need to have access to a pointwise map between \textit{the sample points} on the two shapes. This restricted functional map $\widehat{\*C}$ is a pull-back operator associated to the reduced spaces $\text{Span}\big\{\overline{\phi}_k^\Mm\big\}_k$ and $\text{Span}\big\{\overline{\phi}_k^\Nn\big\}_k$ since both families are orthonormal with respect to $\overline{\*A}$.
Furthermore, using the factorization $\overline{\bm{\Psi}}=\*U\overline{\bm{\Phi}}$ on each shape in~\eqref{eq:Approximated FM} as well as the definition of $\overline{\*A}$, we remark that going from Eq.~\eqref{eq:Approximated FM} to~\eqref{eq:reduced FM} only requires the approximation $\bm{\Pi}\*U^\Mm \simeq \*U^\Nn\overline{\bm{\Pi}}$, for which we will later on derive an upper bound in Proposition~\ref{prop:Approximation Error}.
Note that one might want to replace $\overline{\bm{\Phi}}^\Mm$ by $\overline{\bm{\Psi}}^\Mm$ in Eq.~\eqref{eq:reduced FM} so that the map $\overline{\bm{\Pi}}$ actually transports pointwise values rather than coefficients. In practice as evaluated in Appendix~\ref{app:coefficient correction}, we did not observe any improvement using this modification.

The first benefit of the approximated functional map in Eq.~\eqref{eq:reduced FM} compared to the exact one in Eq.~\eqref{eq:Approximated FM} is the exclusive use of small-sized matrices. Observe that functions $\big(\overline{\phi}_i\big)_i$, are associated with the area  and stiffness matrices $\overline{\*A}$ and $\overline{\*W}$, which define the $L_2$ and $W_1$ inner products, thus allowing to use \textit{all} functional map related algorithms in a straightforward way \emph{without} using any extra approximation or acceleration heuristics.
Eventually a dense pointwise-map between complete shapes can be obtained by identifying the two pull-back operators $\widehat{\*C}$ and $\overline{\*C}$, as described later in Section~\ref{sec:final algorithm}.
As we will see, the resulting correspondences outperform those obtained using remeshed versions of shape and nearest neighbor extrapolation, as our functional map produces sub-sample accuracy.

Secondly, as shown in the following section, our approach is backed by strong theoretical convergence guarantees, providing bounds on approximation errors. In contrast, previous approaches, such as the accelerated version of ZoomOut \cite{melziZoomOutSpectralUpsampling2019} (Sec. 4.2.3) \emph{might not} converge to the true functional maps even when using all available information. 
Namely, Fast ZoomOut indeed samples $q$ points on shapes $\Mm$ and $\Nn$, and approximates $\overline{\*C}$ using
\begin{equation}\label{eq:fast ZoomOut approx}
    {\*C}_{\text{F-ZO}} = \uargmin{\*X} \| \*Q^\Nn \bm{\Psi}^\Nn \*X - {\bm{\Pi}}\*Q^\Mm {\bm{\Psi}}^\Mm \|_F^2
\end{equation}
where $\*Q^\Nn\in\{0,1\}^{q\times n^\Nn}$ with $\*Q^\Nn_{ij}=1$ if and only if $x_j$ is the $i^\text{th}$ sample on shape $\Nn$ (similarly for $\Mm$). Using all points means $\*Q$ matrices are identity. This approximation gives equal importance to all sampled points regardless of their area, and thus fails to converge towards the underlying $\*C$ as the number of samples increases. This means a near uniform sampling strategy is required in practice, which is difficult to achieve on very dense meshes.

In the following section, we provide approximation error bounds for our functional map definition, which we later use to modify the approach from~\cite{nasikunFastApproximationLaplaceBeltrami2018} to reduce these errors and obtain a more accurate and principled correspondence approach.


    \subsection{Approximation Errors}
    \label{sec:approximation bounds}
    Most expressions above involve a given pointwise map $\bm{\Pi}$ between surfaces $\Nn$ and $\Mm$. 
The following lemma provides simple assumptions to obtain a Lipschitz constant for its associated functional map, which will be very useful to derive bounds on the approximation errors of our estimators:
\begin{lemma}\label{lemma:bounded distortion}
Let $\Mm$ and $\Nn$ be compact surfaces and $T:\Nn\to\Mm$ a diffeomorphism.
Then there exists $B_T\in\RR$ so that:
\begin{equation}
    \left\|f\circ T\right\|_{\Nn} \leq B_T\left\|f\right\|_{\Mm} \quad \forall\ f\in L^2(\Mm)
\end{equation}
\end{lemma}
the proof of which can be found in~\cite{huangStabilityFunctionalMaps2018} (Proposition 3.3).


Our overall goal is to use the newly designed functional map $\overline{\*C}$ within a ZoomOut-like functional map estimation algorithm. We therefore expect the approximated functional map to mimic the underlying map $\*C$ when the computed eigenvectors $\overline{\Psi}_k$ approximate well the true ones $\Psi_k$. The following proposition bounds the error between the two functional maps:

\begin{proposition}\label{prop:Approx FM vs GT FM}
Let $\overline{\bm{\Psi}}^\Nn$ (resp. $\overline{\bm{\Psi}}^\Mm$) and $\bm{\Psi}^\Nn$ (resp. $\bm{\Psi}^\Mm$) the approximated and true first $K$ eigenvectors of the Laplacian on $\Nn$ (resp. $\Mm$).
Let $\*C$ and $\overline{\*C}$ be the original and reduced (see Eq.~\eqref{eq:Approximated FM}) functional maps of size $K$, associated to the map $T$. Suppose that $T$ is a diffeomorphism, and let $B_T$ be the bound given by Lemma~\ref{lemma:bounded distortion}.\\
If there exists $\varepsilon\in\RR_+^*$ so that for any $j\in\{1,\dots,K\}$ :
\begin{equation*}
    \|\bm{\Psi}^\Nn_j - \overline{\bm{\Psi}}^\Nn_j\|_\infty \ \leq \varepsilon \text{ and } \|\bm{\Psi}^\Mm_j - \overline{\bm{\Psi}}^\Mm_j\|_\infty \ \leq \varepsilon
\end{equation*}
Then:
\begin{equation}
    \frac{1}{K}\left\| \*C - \overline{\*C}\right\|_2^2 \leq \varepsilon^2 \inpar{1+B_T^2}
\end{equation}
\end{proposition}
The proof can be found in Appendix~\ref{app:proof:Approx FM vs GT FM}. This proposition ensures that a good estimation of the spectrum implies an accurate functional map approximation, and thus its good behavior within matching algorithms.

A more fundamental error to control is the estimation error between the functional maps $\overline{\*C}$ from Def.~\ref{def:Functional map on new space} and $\widehat{\*C}$ from Def.~\ref{def:reduced FM}. As mentioned above, the estimation relies on the identification $\bm{\Pi}\overline{\bm{\Psi}}^\Mm \simeq \*U^\Nn\overline{\bm{\Pi}}\ \overline{\bm{\Phi}}^\Mm$, where $\overline{\bm{\Pi}}$ is a map between the two sets of samples $\Ss^\Nn$ and $\Ss^\Mm$, which we expect to be similar to $\bm{\Pi}$ on these spaces. This approximation treats equivalently the two following procedures: 1) interpolating between values on $\Ss^\Mm$ then transferring using the map $\bm{\Pi}$, 2) transferring values on $\Ss^\Mm$ to values on $\Ss^\Nn$ using $\overline{\bm{\Pi}}$ and then interpolating on $\Nn$.
The following proposition bounds the error of this approximation:
\begin{proposition}\label{prop:Approximation Error}
Let $T:\Nn\to\Mm$ be a pointwise map between the shapes represented by $\bm{\Pi}$, and let $B_T$ be the bound given by Lemma~\ref{lemma:bounded distortion}. Suppose that $T_{\vert\Ss^\Nn}:\Ss^\Nn\to\Ss^\Mm$ is represented by $\overline{\bm{\Pi}}$.\\
Let $\alpha=\min_j u_j^\Mm(v_j)\in[0,1]$. 
Suppose further that there exists $\eps>0$ so that for any $k\in\{1,\dots,K\}$ and $x,y\in\Ss^\Mm$:
\begin{equation}
    d^\Mm(x,y) \leq \rho^\Mm \Rightarrow | \overline{\Psi}^\Mm_k(x) - \overline{\Psi}^\Mm_k(y)|\leq \eps
\end{equation}
and
\begin{equation}
    \ d^\Mm(x,y) \leq \rho^\Mm \Rightarrow | \overline{\Phi}^\Mm_k(x) - \overline{\Phi}^\Mm_k(y)|\leq \eps.
\end{equation}
Then
\begin{equation}
    \frac{1}{K}\left\|\bm{\Pi}\overline{\bm{\Psi}}^\Mm - \*U^\Nn\overline{\bm{\Pi}}\ \overline{\bm{\Phi}}^\Mm \right\|_\Nn^2 \leq \eps^2(1-\al) +\eps^2B_T^2
\end{equation}
\end{proposition}
The proof is given in Appendix~\ref{app:proof:Approximation Error}. This proposition shows that the estimation error depends on two parameters: 1) the variation $\varepsilon$ of the eigenfunctions w.r.t to the sample distance $\rho$, 2) the \emph{self-weights} $u_j(v_j)$ from the local functions defined in the basis approximation.
Note that since the basis functions $u_j$ verify $0\leq u_j\leq 1$ and satisfy the partition of unity, they can be interpreted as interpolation weights from values at sampled points to values on the entire shape. This makes the dependence in $\alpha$ more intuitive, as our approximation relies on the local identification of basis coefficients with function values.
A discussion on the numerical values of the quantities used in Proposition~\ref{prop:Approximation Error} are provided in Appendix~\ref{app:parameters values}.

In the following, we will therefore seek to modify the basis approximation~\cite{nasikunFastApproximationLaplaceBeltrami2018} in order to maximize $\al$ while retaining both the quality of the approximation of the true Laplacian spectrum, necessary to apply functional maps-related algorithms.

\begin{figure}
    \centering
    \includegraphics[width=.9\linewidth]{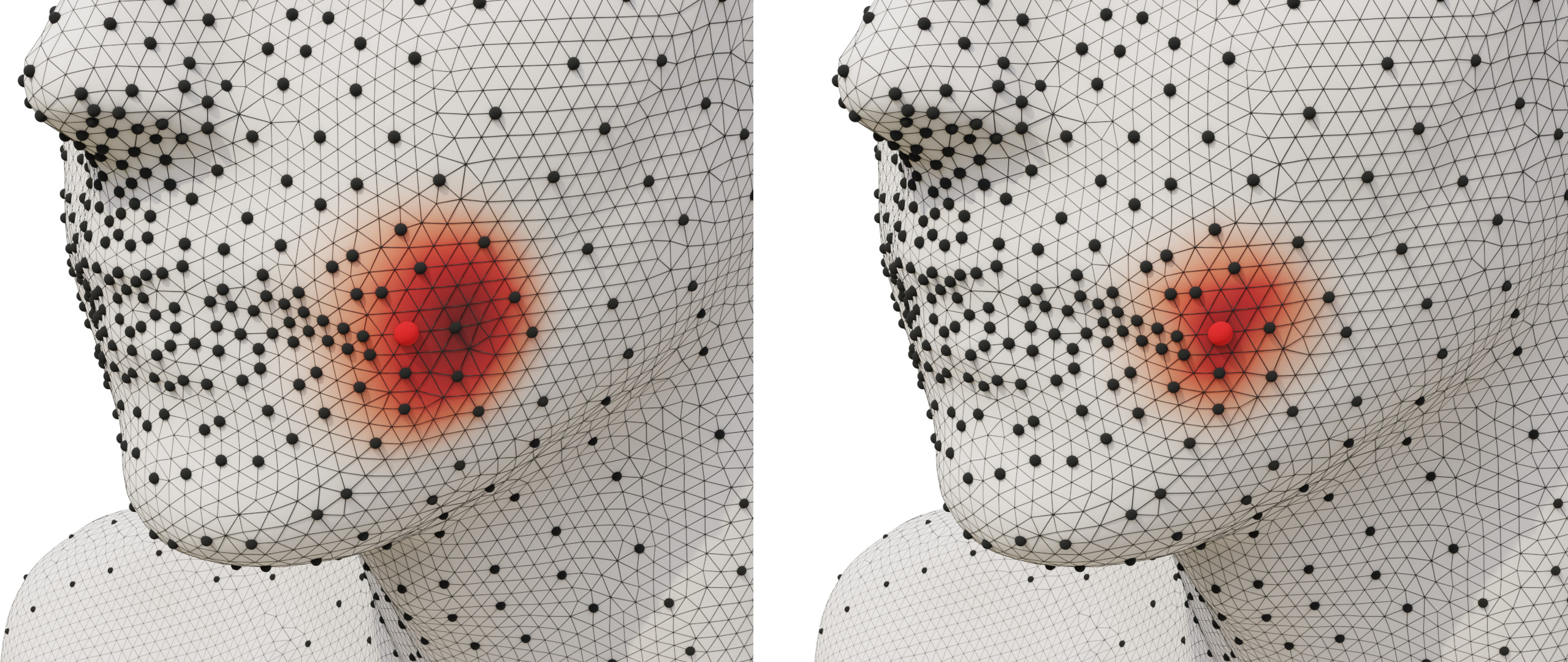}
    \caption{Example of a local function $u_j$ (red color) centered on $v_j$ (red vertex),  visualized without (Left) and with (Right) our adaptive radius strategy. Other samples $v_k$ are displayed in black.
    }
    \label{fig:adaptive radius}
\end{figure}

    \subsection{Improved Eigenbasis Approximation}
    \label{sec:new eigenbasis}

In this section, we propose a modification of the approach from~\cite{nasikunFastApproximationLaplaceBeltrami2018}, based on the theoretical bounds introduced above. For the rest of this section, we focus on a single shape, as the basis computations are done on each shape independently. 

As seen from Prop.~\ref{prop:Approximation Error}, high self weights allow to stabilize our functional map approximation.
Interestingly with the construction in \cite{nasikunFastApproximationLaplaceBeltrami2018}, the value $u_j(v_j)$  only depends on the geodesic distance between $v_j$ and other sampled points $v_i$ for $i\neq j$:
\begin{equation} \label{eq:self weight}
    u_j(v_j) = \frac{1}{1+\sum_{i\neq j}\tilde{u}_i(v_j)}.
\end{equation}
where $\tilde{u}_i$ are the un-normalized local functions.
We modify the pipeline from~\cite{nasikunFastApproximationLaplaceBeltrami2018} in order to increase these values as follows: we first define a per-sample radius $\rho_j$ for $j\in\{1,\dots,p\}$ instead of a single global value $\rho$.
Given a sample point $v_j$ with a small self-weight $u_j(v_j)$, radius $\rho_j$ is kept untouched as it has no influence on the self-weight, but we instead reduce the radius $\rho_i$ of its most influential neighbor - that is the radius of the point $v_i$ with the highest value $\tilde{u}_i(v_j)$.
Following Eq.~\eqref{eq:self weight} this eventually increases the value of $u_j(v_j)$. 
Note that this modification doesn't change the value $u_i(v_i)$ and increase the self weights of all its neighbors. This way all self weights are non-decreasing during the algorithm, with at least one of them increasing.
This extra adaptation additionally comes at a negligible computational cost as it only requires re-evaluating $u_j$ at a set of fixed vertices. In particular, this does not require additional local geodesic distance computations. More details are provided in Sec.~\ref{sec:implementation}, and the algorithm to compute these new functions is displayed in Algorithm~\ref{alg:adaptive algorithm}.
\begin{figure}
    \centering
    \includegraphics[width=.9\linewidth]{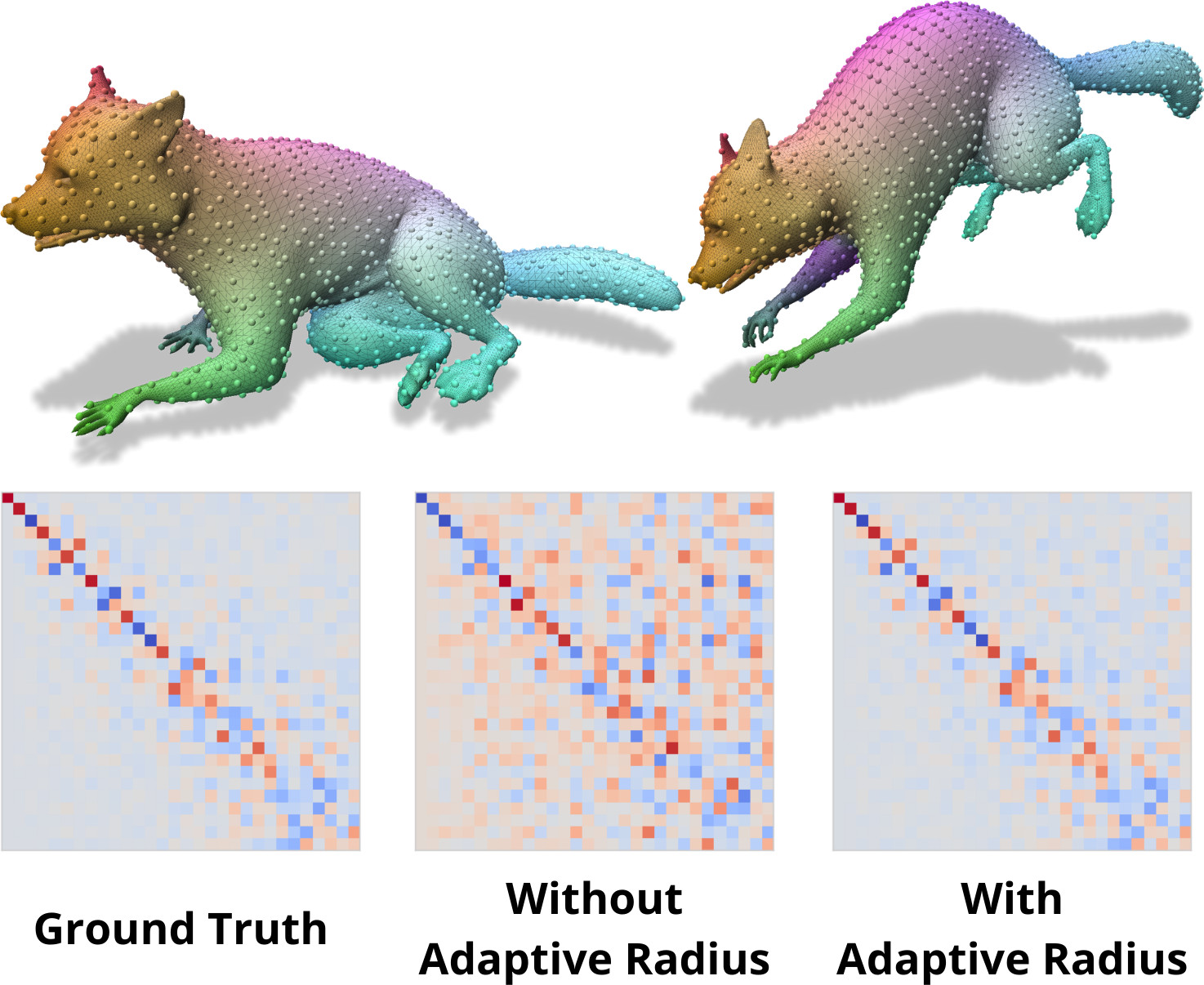}
    \caption{\label{fig:FM vs reduced FM}Effect of the adaptive radius on functional map approximation. Top row displays a pointwise map $T$ from the right mesh to the left mesh using color transfer. Bottom row displays $\overline{\*C}$ (Left), $\widehat{\*C}$ when using the pipeline from~\cite{nasikunFastApproximationLaplaceBeltrami2018} (Middle) and our functional map $\widehat{\*C}$ (Right).}
\end{figure}
We observe that the adaptive radius strategy generates better local functions than those introduced in~\cite{nasikunFastApproximationLaplaceBeltrami2018}, especially for non-uniform sampling, as can be seen on a surface from the DFaust dataset~\cite{bogoDynamicFAUSTRegistering2017} in Figure~\ref{fig:adaptive radius}.
Note that since we focus on \textit{local} analysis, a desirable property of the local interpolation function is the consistency across different shapes when only values at the samples are provided. With a single global radius, we see on Figure~\ref{fig:adaptive radius} that these functions can be heavily distorted by the normalization procedure, which is corrected by our approach. However, increasing the self-weights too close to $1$ also deteriorates the results, as any vertex $x$ within the radius of a single sample will be given the value of the sample point. There thus exists a limit at which this procedure ceases to be helpful, and the only solution then lies in increasing the number of samples on the shape.

\begin{algorithm}
\caption{Computation of local functions with adaptive radius}\label{alg:adaptive algorithm}
\begin{algorithmic}[1]
\Require Mesh $\Mm$, samples $(v_k)_k$, initial $\rho_0$, threshold $\varepsilon$
\State $\rho_j \gets \rho_0\quad\forall j$
\State Compute local functions $U$ with radius $\rho$~:\ \eqref{eq:tilde u def},\ \eqref{eq:normalization of u}
\State Add sample points if necessary
\While{some $k$ with $u_k(v_k) < \varepsilon$} 
    \State $j \gets \uargmax{i\neq k}{u_i(v_k)}$ 
    \State $\rho_j \gets \rho_j/2$
    \State update all $u$ using Eq.~\eqref{eq:tilde u def},\ \eqref{eq:normalization of u}
\EndWhile
\State Add unseen vertices in the sample

\end{algorithmic}
\end{algorithm}

The positive effect of our adaptive radius algorithm for functional map estimation is further visualized in Figure~\ref{fig:FM vs reduced FM}, where given a single pointwise map $T$, 
we display the exact functional map on the approximated spaces $\overline{\*C}$, and two approximated functional maps $\widehat{\*C}$, one being computed with a shared radius~\cite{nasikunFastApproximationLaplaceBeltrami2018} and the other with our adaptive radius scheme.
We highlight that the ground truth functional map actually differ for each approximation $\widehat{\*C}$ as the reduced functional spaces are modified, which makes values not directly comparable. However, we observe that the two ground truth maps have nearly identical sparsity structure  (see Appendix~\ref{app:FM approx}), which is why we only display one in Figure~\ref{fig:FM vs reduced FM}. Note that using the adaptive radius strategy then generates a sparsity pattern on matrix $\widehat{\*C}$ very close to the ground truth one.


    \subsection{Scalable ZoomOut}
    \label{sec:final algorithm}
    In light of the previous discussions and theoretical analysis, we now describe how to use the approximated functional map $\widehat{\*C}$ within a standard ZoomOut pipeline~\cite{melziZoomOutSpectralUpsampling2019}. Our complete pipeline is summarized in Algorithm~\ref{alg:complete algo}, where the notation $\overline{\bm{\Phi}}_{1:k}$ indicates that we only use the first $k$ column of matrix $\overline{\bm{\Phi}}_{1:k}$. 


\begin{algorithm}
\caption{Scalable ZoomOut}\label{alg:complete algo}
\begin{algorithmic}[1]
\Require Meshes $\Mm$ and $\Nn$, threshold $\varepsilon$, initial map
\State Sample $\Ss^\Mm$ and $\Ss^\Nn$ using Poisson Disk Sampling
\State Compute $\*U^\Mm$ and $\*U^\Nn$ using Algo.~\ref{alg:adaptive algorithm}
\State Approximate eigenvectors $\overline{\bm{\Phi}}^\Mm$ and $\overline{\bm{\Phi}}^\Mm$ solving~\eqref{eq:reduced LBO decomposition}
\State Set $\overline{\bm{\Psi}}^\Mm = \*U^\Mm \overline{\bm{\Phi}}^\Mm$ and $\overline{\bm{\Psi}}^\Nn = \*U^\Nn \overline{\bm{\Phi}}^\Nn$
\State Obtain $\overline{\bm{\Pi}}$ between \textit{samples} using the initial map
\For{$k=k_{\text{init}}$:$k_\text{final}$}
\State $\widehat{\*C} = \inpar{\overline{\bm{\Phi}}^\Nn_{1:k}}^\top \overline{\*A}^\Nn \overline{\bm{\Pi}}\ \overline{\bm{\Phi}}^\Mm_{1:k}$
\State $\overline{\bm{\Pi}} = \text{NNsearch}\big(\overline{\bm{\Phi}}^\Mm_{1:k}, \overline{\bm{\Phi}}^\Nn_{1:k}\widehat{\*C}\big)$ potentially using~\eqref{eq:possible images}
\EndFor
\State $\bm{\Pi} = \text{NNsearch}\big(\overline{\bm{\Psi}}^\Mm_{1:k}, \overline{\bm{\Psi}}^\Nn_{1:k}\widehat{\*C}\big)$
\State \textbf{Return} $\bm{\Pi}$
\end{algorithmic}
\end{algorithm}

As mentioned earlier, using $\widehat{\*C}$ and matrices $\overline{\*A}$ and $\overline{\*W}$ allows to apply the ZoomOut algorithm directly, as if it was applied on remeshed versions of the shapes with only $p$ vertices. This results in a refined functional map $\widehat{\*C}^*$ and a refined pointwise map \emph{between samples} $\overline{\bm{\Pi}}^*$. The last remaining non-trivial task consists in converting the refined functional map into a global pointwise map $\bm{\Pi}^*$ between the original dense meshes.

Standard approaches using remeshed versions of the shapes extend maps via nearest neighbors,  resulting in locally constant maps. Instead, we identify $\widehat{\*C}$ and $\overline{\*C}$, which then allows us to compute the pointwise map $\bm{\Pi}^*$ by solving the standard least square problem:
\begin{equation}
    \bm{\Pi}^* = \uargmin{\bm{\Pi}}{\|\overline{\bm{\Psi}}^\Nn\widehat{\*C}^* - \bm{\Pi} \overline{\bm{\Psi}}^\Mm\|_{\*A^\Nn}^2}.
\end{equation}
Since $\*A$ is diagonal this problem reduces to a nearest neighbor search for each vertex $x\in\Nn$. This way, the obtained pointwise map is no longer locally constant which results in a significant gain of quality with respect to typical approaches.

On meshes containing millions of vertices, this nearest neighbor search can, however, still be very slow. In these cases, we propose to use the computed pointwise map $\overline{\bm{\Pi}}$ as a guide to reduce the search space as follows: for $x\in\Nn$, we first select the indices of its nearest sample points $N(x)=\{j\ |\ u_j^\Nn(x) > 0\}$, and create the set of possible \emph{images} as the points in $\Mm$ close to the image of this set under the map $\overline{\bm{\Pi}}$, that is
\begin{equation}\label{eq:possible images}
    \Ii(x) = \{y\ |\ \exists j\in N(x),\ u_{\bar{T}(j)}(y)>0 \}
\end{equation}
where $\bar{T}$ is the function representation of $\overline{\bm{\Pi}}$. Since local functions $u_j$ are compactly supported, in practice, they are stored as sparse vectors and extracting the set of possible images of a given vertex therefore can be done efficiently through simple indexing queries.


    \subsection{Implementation}
    \label{sec:implementation}
    We implement the complete algorithm in Python and provide the code at \URL{https://github.com/RobinMagnet/Scalable_FM}.

Following~\cite{nasikunFastApproximationLaplaceBeltrami2018}, we generate sparse samples $\Ss$ using Poisson Disk sampling, and run a fixed-radius Dijkstra algorithm starting at all sampled points $v_j$ to build local functions $u_j$. Values can be stored in a sparse $n\times p$ matrix where $p$ is the number of samples. Note that the adaptive radius algorithm presented in Section~\ref{sec:new eigenbasis} does not require additional geodesic distance computations. Furthermore finding the set of potential images for a point as described in Section~\ref{sec:final algorithm} simply reduces to checking non-zero indices in a sparse matrix. More details are provided in Appendix~\ref{app:implementation}.

\begin{figure*}[tbp]
  \centering
  \mbox{} \hfill
  \includegraphics[width=.8\linewidth]{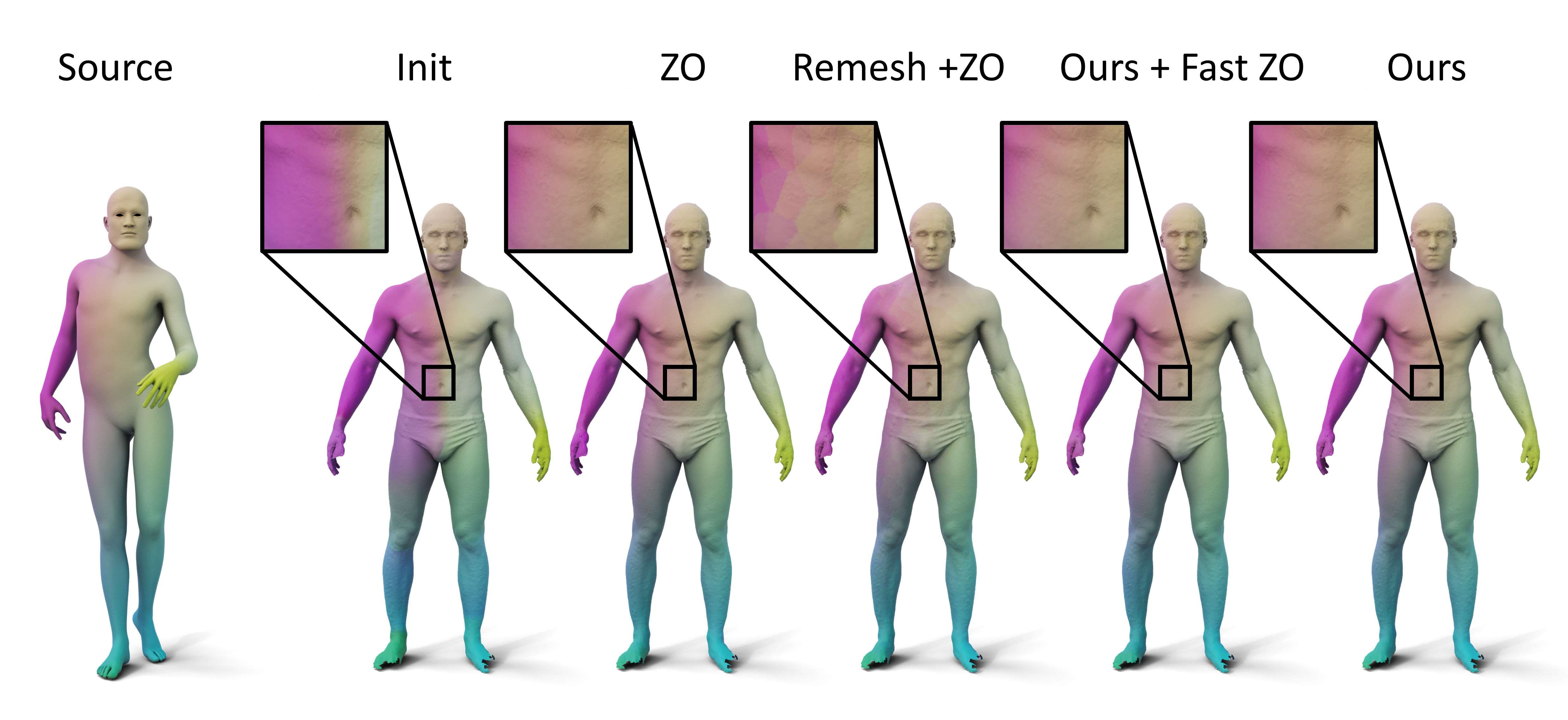}
  \hfill \mbox{}
  \caption{\label{fig:qualitative_shrec}%
Qualitative results on the SHREC19 dataset. Although processing time differ heavily, there is no significant difference between our method and results from ZoomOut. However, remeshing the surface before ZoomOut results in locally constant correspondences.}
\end{figure*}
\section{Results}
\label{sec:results}
In this section we evaluate our method, while focusing on two aspects. Firstly we verify that our method outperforms existing approaches in terms of speed at all steps of the pipeline - that is pre-processing as well as the ZoomOut algorithm. Secondly we show this gain in speed comes at a minimal cost in terms of quantitative metrics. In particular we verify that although our pipeline relies on sparse samples, we eventually obtain clear sub-sample accuracy in the correspondences.

\begin{table}[]
\centering
\caption{Timing in seconds for different methods when processing a pair with $50$K and $200$K vertices and applying ZoomOut from spectral size $20$ to $100$}
\label{tab:timing SHREC19}
\footnotesize
{\def\arraystretch{1}\tabcolsep=0.3em
\begin{tabular}{ccccc|c}\toprule[0.8pt]
methods & \bfseries\itshape Preprocess & \bfseries\itshape LBO & \bfseries\itshape ZoomOut & \bfseries\itshape Conversion  &\bfseries\itshape Total (s)\\ \midrule[0.8pt]
ZO & $1$ & $132$ & $410$ & $83$ & $626$\\
 Fast ZO & $10$ & $132$ & $1$ & $44$ & $187$\\
 R + ZO & $14$ & $2$ & $3$ & $1$ & $21$\\ \midrule[0.8pt]
Ours & $10$ & $7$ & $5$ & $44$ & $65$\\ \bottomrule[0.8pt]
\end{tabular}}
\end{table}

\subsection{Timings}
The method introduced in~\cite{nasikunFastApproximationLaplaceBeltrami2018} aimed at approximating the spectrum of the Laplace Beltrami Operator efficiently. As mentioned above, the additional building blocks we introduced in Section~\ref{sec:new eigenbasis} come at a nearly negligible computational cost as the main bottleneck lies in local geodesic distances computations, which are not recomputed.
The main benefit of our method appears when considering the processing time of the ZoomOut algorithm. Indeed since our algorithm does not involve any $n$-dimensional matrices, its running time becomes entirely agnostic to the original number of vertices. Only the final conversion step, which converts the refined functional map into a dense point-wise map, scales with the number of vertices. Table~\ref{tab:timing SHREC19} displays an example of timings when applying the ZoomOut algorithm between two meshes with respectively $50$ and $200$ thousands vertices. We compare the standard ZoomOut algorithm (ZO), the Fast ZoomOut algorithm (Fast ZO), the standard ZoomOut applied to remeshed versions of the shapes with nearest neighbor extrapolation (R+ZO) and our complete pipeline with $p=3000$ samples on each shape. Notice that farthest point sampling used in Fast ZoomOut can become quite slow on dense meshes compared to Poisson sampling, which explains the similar preprocessing timings between our method and Fast ZoomOut.



\subsection{Evaluation}

\paragraph*{Dataset} As most shape matching methods scale poorly with the number of vertices, there are few benchmarks with dense meshes and ground truth correspondences for evaluation.
The SHREC19 dataset \cite{melzi2019shrec} consists of $430$ pairs of human shapes with different connectivity, all of which come with initial correspondences. Meshes in this dataset have on average $38\ 000$ vertices, with the smallest and largest number of vertices having respectively $4700$ and $200\ 000$ vertices. Due to the limitations of existing shape matching methods, a remeshed version of this dataset is commonly used. In contrast, we display results on the \textit{complete dense dataset}, and show that our method obtains similar results as ZoomOut \cite{melziZoomOutSpectralUpsampling2019} in only a fraction of the required time.

\paragraph*{Metrics} We evaluate different methods using standard metrics~\cite{renContinuousOrientationpreservingCorrespondences2019} for dense shape correspondence, that is accuracy, coverage and smoothness. The accuracy of a computed dense map $T:\Nn\to\Mm$ gives the average geodesic distance between $T(x)$ and $T^*(x)$ for all $x\in\Nn$ where $T^*$ denotes the ground truth map. Note that since maps on SHREC19 are only evaluated on a small subset of $6890$ points this metric only captures partial information, and locally constant maps can still achieve high accuracy.
Coverage and smoothness metrics provide additional information on the quality of correspondences and are sensitive to locally constant correspondences. Coverage is defined as the ratio of area covered by the pointwise map, and smoothness is the Dirichlet energy defined as the squared $L^2$ norm of the gradient of the transferred coordinates.


\begin{table}[]
\centering
\caption{Evaluation of different methods on the complete SHREC19 dataset. Blue highlights the best two methods.}
\label{tab:quantitative evaluation SHREC}
{\def\arraystretch{1}\tabcolsep=0.3em
\begin{tabular}{cccc}\toprule[0.8pt]
methods & \bfseries\itshape Accuracy & \bfseries\itshape Coverage & \bfseries\itshape Smoothness\\ \midrule[0.8pt]
    Init & $60.18$ & $26.5\ \%$ & $9.5$\\ \midrule[0.5pt]
    GT & $-$ & $33.0\ \%$ & $10.43$\\
    ZO & \cellcolor{tabblue!20}$\mathbf{26.84}$ & \cellcolor{tabblue!20}$\mathbf{61.5}\ \%$ & \cellcolor{tabblue!20}$6.2$\\
    R + ZO & $28.57$ & $18.0\ \%$ & $15.0$\\
    Ours w/o radius & $71.35$ & $29\ \%$ & 52.2 \\
    Ours + Fast ZO & $29.5$ & \cellcolor{tabblue!20}$59.7$ \% & $6.4$\\ \midrule[0.8pt]
    Ours & \cellcolor{tabblue!20}$27.78$ & $56.7\ \%$ &\cellcolor{tabblue!20} $\mathbf{5.6}$\\ \bottomrule[0.8pt]
\end{tabular}}
\end{table}

\paragraph*{ZoomOut} We compare our method (Ours) using $3000$ samples first to the same algorithm without adaptive radius (Ours w/o radius), to the standard ZoomOut~\cite{melziZoomOutSpectralUpsampling2019} algorithm applied on the dense meshes (ZO) and on remeshed versions with $3000$ vertices (R+ZO). 
We don't compare to other standard shape matching baseline~\cite{eisenbergerSmoothShellsMultiScale2020, ezuzReversibleHarmonicMaps2019} first since we only wish to approximate results from ZoomOut, but also because these baselines don't scale to high number of vertices.
Additionally, despite the lack of theoretical guarantees, we evaluate a new version of Fast ZoomOut which uses functional map approximation~\eqref{eq:fast ZoomOut approx} on the approximated functional space $\overline{\Ff}$ introduced in section~\ref{sec:new FM} (Ours + Fast ZO). 
Table~\ref{tab:quantitative evaluation SHREC} shows the values of the evaluation metrics on the SHREC19 dataset where the accuracy curves can be found on Figure~\ref{fig:SHREC accuracy curves}, and Figure~\ref{fig:qualitative_shrec} shows an example of a map computed on two dense meshes. 
We see that all methods but R+ZO produce similar metrics although processing times vary significantly. In contrast, the fastest method R+ZO produces locally constant maps as seen on Figure~\ref{fig:qualitative_shrec}, which results in poor coverage and smoothness metrics.
While our results are similar to ZoomOut and Fast ZoomOut, we stress that our results were obtained at a fraction of the processing time of ZoomOut, and come with theoretical upper bounds and control parameters on approximations which Fast ZoomOut does not have.

\begin{figure}
    \centering
    \includegraphics[width=.7\linewidth]{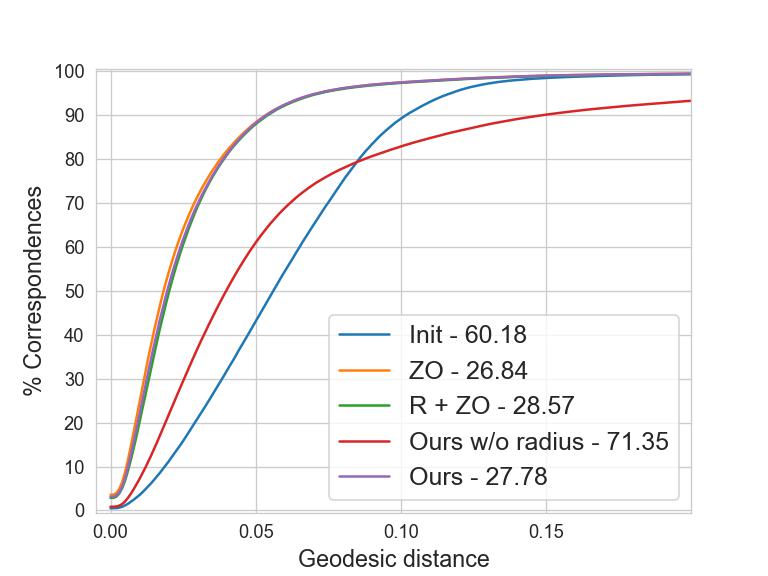}
    \caption{Accuracy curves for different methods presented in Table~\ref{tab:quantitative evaluation SHREC}. Numbers in the legend provide the average geodesic error ($\times10^3$).}
    \label{fig:SHREC accuracy curves}
\end{figure}

\paragraph*{Sub-sample accuracy} One Figure~\ref{fig:sub-sample accuracy}, we provide a result using texture transfer after applying our scalable ZoomOut on a pair of real scans of humerus bones obtained using a CT scanner~\cite{rustamovMapbasedExplorationIntrinsic2013}.
This Figure shows how our algorithm obtain sub-sample accuracy, as the transferred texture remains smooth even though samples are quite sparse on each shape. We display similar results using texture transfer on the SHREC19 dataset on Figure~\ref{fig:texture shrec} and in Appendix~\ref{app:texture transfer}, which provides further details on the shapes.

\paragraph*{Adaptive Radius} While results on Table~\ref{tab:quantitative evaluation SHREC} highlight the efficiency of the adaptive radius scheme, we additionally evaluate how this heuristic allows to improve the estimation $\Delta = \|\overline{\*C} -\widehat{\*C}\|$ presented in section~\ref{sec:contributions}. For this we simply compute $\overline{\*C}$ and $\widehat{\*C}$ with $K=20$ for all initial maps of the SHREC19 dataset, and evaluate the norms of the estimation errors $\Delta$ which we provide in Table~\ref{tab:quantitative evaluation radius}. In this experiment we notice our method improves the baseline by two orders of magnitude.

\begin{table}[h]
\centering
\caption{Norm of the estimation error $\Delta$ with and without adaptive radius on the SHREC19 dataset}
\label{tab:quantitative evaluation radius}
\footnotesize
{\def\arraystretch{1}\tabcolsep=0.3em
\begin{tabular}{ccc}\toprule[0.8pt]
& \bfseries\itshape w/o radius & \bfseries\itshape Ours \\ \midrule[0.8pt]
$\Delta$ ($\times10$) & $1.486$ & $0.018$ \\
\bottomrule[0.8pt]
\end{tabular}}
\end{table}



\section{Conclusion, Limitations and Future Work}

In this paper we introduced a new scalable approach for computing correspondences between non-rigid shapes, represented as possibly very dense meshes. Our method is based on the efficient approach for estimating the Laplace-Beltrami eigenbasis \cite{nasikunFastApproximationLaplaceBeltrami2018} using optimization of coefficients of local extension functions built from a sparse set of samples. Key to our approach is careful analysis of the relation between functional spaces on the samples and those on the original dense shapes. For this, we extend this approach proposed in \cite{nasikunFastApproximationLaplaceBeltrami2018} and demonstrate how better behaved local functions can be obtained with very little additional effort. We use this construction to define a functional map approximation that only relies on information stored at the samples, and provide theoretical guarantees for this construction. Finally, we use these insights to propose a scalable variant of the ZoomOut algorithm \cite{melziZoomOutSpectralUpsampling2019}, which allows to compute high-quality functional and point-to-point maps between very dense meshes at the fraction of the cost of the standard approach.

\begin{figure}[htb]
  \centering
  \includegraphics[width=.9\linewidth]{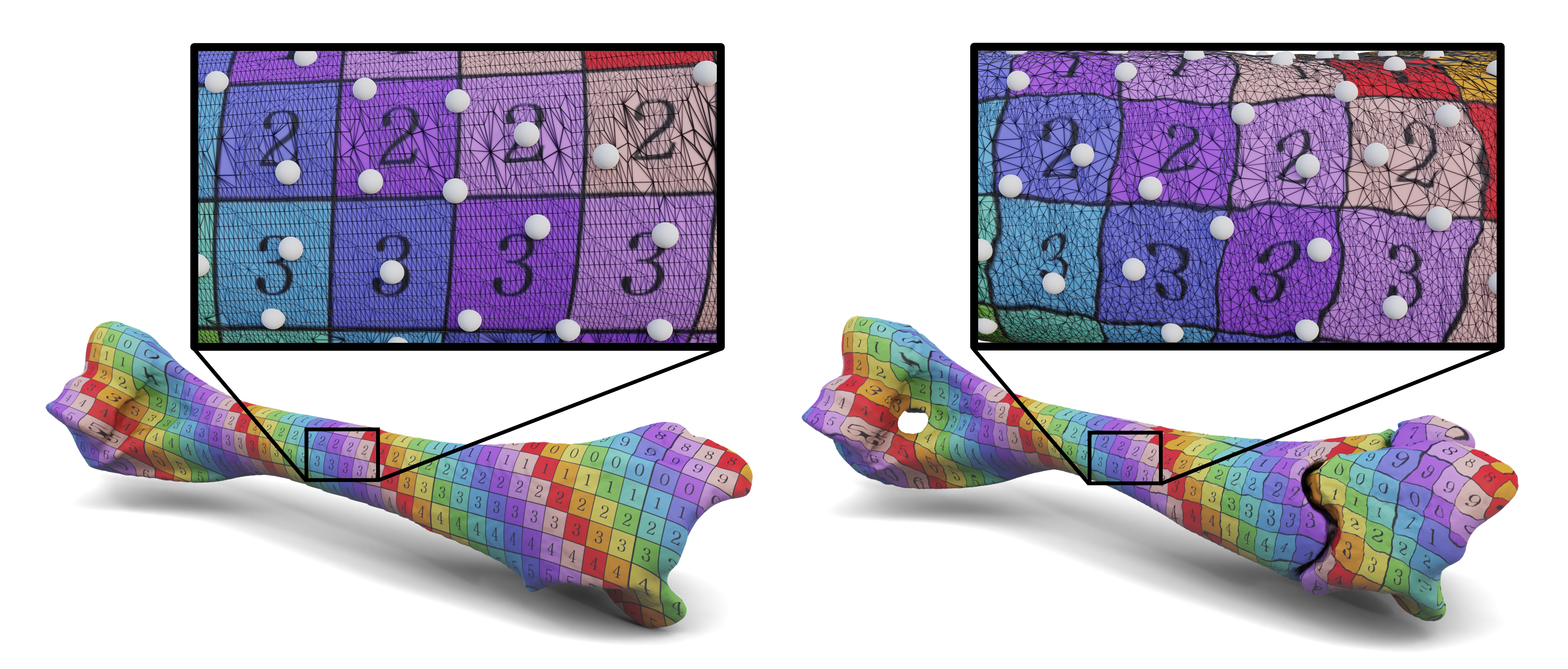}

  %
  %
  \caption{\label{fig:sub-sample accuracy}
           Texture transfer using our scalable version of ZoomOut. Samples used in the pipeline are shown as white dots.}
\end{figure}

Although our method achieves high-quality results, it still has several limitations. First, it relies heavily on the mesh structure, and is not directly applicable to other representations, such as point clouds. Second, our method depends on a critical hyperparameter, which is the number of samples. We have observed that 3000 samples perform well on a very wide range of settings, but it would be interesting to investigate the optimal number, depending on the size of the spectral basis. Furthermore, we use Poisson sampling as advocated in \cite{nasikunFastApproximationLaplaceBeltrami2018}, which gives good results in practice. However, the optimal choice of the sampling procedure, depending on the geometric properties of shapes under consideration, would be an equally interesting venue for investigation. Lastly, an out-of-core implementation, capable of handling meshes with 10’s of millions to billions of vertices, while possible in principle, would be an excellent practical future extension of our approach.

\begin{figure}[htb]
  \centering
  \includegraphics[width=.9\linewidth]{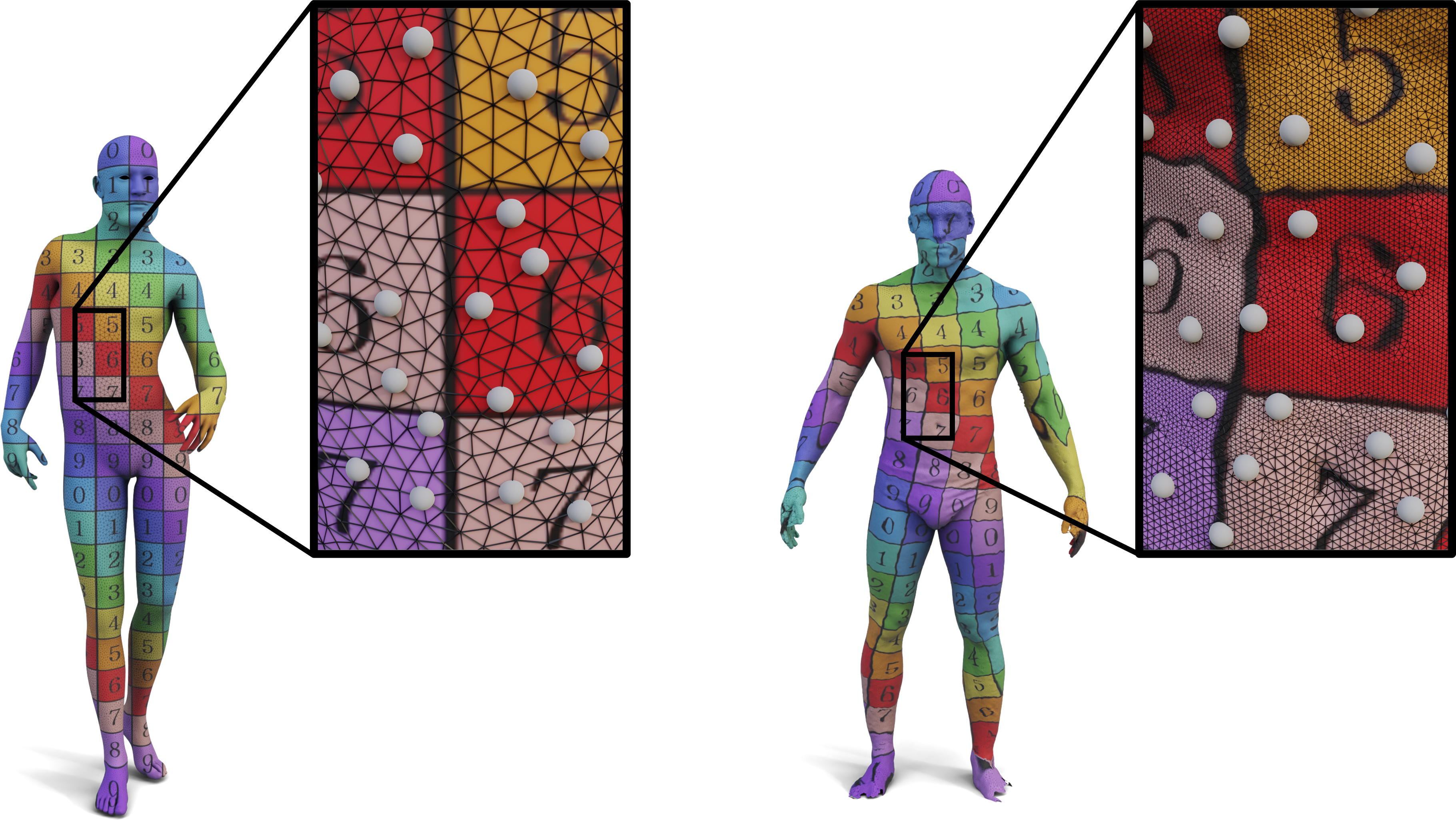}
  \caption{\label{fig:texture shrec}
           Texture transfer using our scalable version of ZoomOut on a pair of the SHREC19 dataset. Samples used in the pipeline are shown as white dots.}
\end{figure}

\paragraph*{Acknowledgments} The authors thank the anonymous reviewers for their valuable comments and suggestions. Parts of this work were supported by the ERC Starting Grant No. 758800 (EXPROTEA) and the ANR AI Chair AIGRETTE.

\bibliographystyle{eg-alpha-doi} 
\bibliography{biblio}       

\appendix
\section{Function $\chi$}
\label{app:chi function}

The function $\chi:\RR_+\to[0,1]$, differentiable with $\chi(0)=1$ and $\chi(x)=0$ for $x\leq 1$ can be defined two ways. Following~\cite{nasikunFastApproximationLaplaceBeltrami2018}, we use the polynomial interpolation function $\chi:x\mapsto 1 - 3x^2+ 2x^3$.
Another possibility is to use the standard $C^\infty$ compactly supported function $\chi : x\mapsto \exp\inpar{1 - \frac{1}{1 - x^2}}$, which we found not as good as the polynomial interpolation regarding results. Both functions are displayed on Figure~\ref{fig:chi func}.

\begin{figure}[htb]
    \centering
    \includegraphics[width=.4\textwidth]{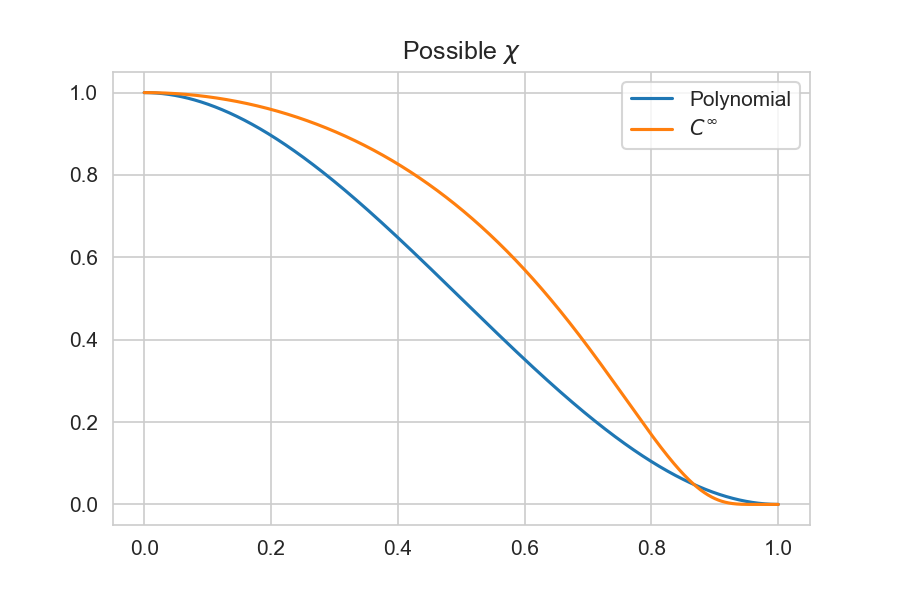}
    \caption{Possible choices for function $\chi$}
    \label{fig:chi func}
\end{figure}

\section{Coefficient weighting}
\label{app:coefficient correction}
Table~\ref{tab:coef reweight eval} compares our algorithm (Ours) with a similar one (Ours + reweight), where we replace $\overline{\bm{\Phi}}^\Mm$ by $\overline{\bm{\Psi}}^\Mm$ in Eq.~\eqref{eq:reduced FM} so that the map $\overline{\bm{\Pi}}$ actually transports pointwise values rather than coefficient, as mentioned in section~\ref{sec:new FM}. This amounts to reweigthing local coefficients to actually become function values.
We notice this method does not improve our pipeline on the SHREC19 dataset.

\begin{table}[]
\centering
\caption{Evaluation of the reweighting scheme on the SHREC19 dataset.}
\label{tab:coef reweight eval}
{\def\arraystretch{1}\tabcolsep=0.3em
\begin{tabular}{cccc}\toprule[0.8pt]
methods & \bfseries\itshape Accuracy & \bfseries\itshape Coverage & \bfseries\itshape Smoothness\\ \midrule[0.8pt]
    Init & $60.18$ & $26.5\ \%$ & $9.5$\\ \midrule[0.5pt]
    Ours + Reweight & $28.1$ & $54.6\ \%$ & $6.3$ \\
    Ours & $27.78$ & $56.7\ \%$ & $5.6$\\ \bottomrule[0.8pt]
\end{tabular}}
\end{table}

\section{Proof of Proposition~\ref{prop:Approx FM vs GT FM}}
\label{app:proof:Approx FM vs GT FM}








\begin{proof}
We first note that the entries of the functional maps $\*C$ and $\overline{\*C}$ can be written
\begin{align}
    \*C_{i,j} &= \langle \psi^\Nn_j, \psi_i^\Mm \circ T\rangle_\Nn\\
    \overline{\*C}_{i,j}& = \langle \overline{\psi}^\Nn_j, \overline{\psi}_i^\Mm \circ T\rangle_\Nn
\end{align}

Furthermore, given $f_1$, $g_1$, $f_2$, $g_2$ functions on $\Nn$.

Then for any $x\in\Nn$
\begin{equation}
\begin{split}
    f_1(x)g_1(x) - f_2(x)g_2(x) = &f_1(x)\inpar{g_1(x)-g_2(x)} \\
    &+ g_2(x)\inpar{f_1(x)-f_2(x)}
\end{split}
\end{equation}

With $f_1 = \psi^\Nn_i$, $g_1 = \psi^\Mm_j\circ T$, $f_2 =\overline{\psi}^\Nn_j$ and $g_2 = \overline{\psi}_i^\Mm \circ T$, we have by hypothesis
\begin{equation}
\begin{split}
    \left\|f_2 - f_1\right\|_\infty \ \leq \varepsilon \\ 
    \left\|g_2 - g_1\right\|_\infty \ \leq \varepsilon \\
    \left\|g_2\right\|_\Nn \leq B_T
\end{split}
\end{equation}

Therefore 
\begin{align*}
    \left| \langle f_1, g_1 \rangle_\Nn - \langle f_2, g_2 \rangle_\Nn \right|^2 &=
    \left|\int_\Nn \left(f_1(x)g_1(x) - f_2(x)g_2(x)\right) d\mu^\Nn(x)\right|^2 \\
    &\leq \int_\Nn f_1(x)^2\inpar{g_1(x)-g_2(x)}^2 d\mu^\Nn(x) \\ &+ \int_\Nn g_2(x)^2\inpar{f_1(x)-f_2(x)}^2 d\mu^\Nn(x) \\
    &\leq \varepsilon \int_\Nn f_1(x)^2d\mu^\Nn(x) + \varepsilon \int_\Nn g_2(x)^2d\mu^\Nn(x) \\
    &\leq \varepsilon\left(\|f_1\|_\Nn^2 + \|g_2\|_\Nn^2 \right)\\
    &\leq \varepsilon^2\left(1+B_T^2\right)
\end{align*}

Summing for all elements of the matrix $\*C$ gives the result.
\end{proof}

\section{Proof of Proposition~\ref{prop:Approximation Error}}
\label{app:proof:Approximation Error}
This proof relies on the following proposition, 
\begin{proposition}\label{prop:interpolation error}Given, $\Mm$ a surface, $\big(v_j\big)_j$ and $\big(u_j\big)_j$ built as described in Sec~\ref{sec:background:approx}.\\
Given $f:\Mm\to\RR$, suppose there exists $\eps>0$ so that for any $x,y\in\Mm$, 
$d(x,y)\leq\rho \implies |f(x) - f(y)|\leq \eps$.

Then the interpolation error between $f$ and $\tilde{f} = \sum_j f(v_j)u_j$ is bounded by $\eps$:
\begin{equation}\label{eq:interpolation error}
    |\tilde{f}(x) - f(x)| \leq \eps \quad \forall x \in \Mm
\end{equation}

And for any $j\in\{1,\dots,p\}$
\begin{equation}\label{eq:interpolation error on sample}
    |\tilde{f}(v_j) - f(v_j)| \leq \eps \inpar{1 - u_j(v_j)}
\end{equation}
\end{proposition}
\begin{proof} \textit{of Prop.~\ref{prop:interpolation error}}

Let $f:\Mm\to\RR$, $\Ss=\left\{v_1,\dots,v_p\right\}$ a sample of $\Mm$ associated to a radius $\rho$.
Since $\inpar{u_j}$ verify $\sum_j u_j=\mathbb{1}$, for any $x\in\Mm$ we have
\begin{equation}
    f(x) = \sum_{j=1}^p u_j(x) f(x)
\end{equation}

Therefore if $\tilde{f}=\sum_j f(v_j)u_j$
\begin{equation}
\begin{split}
    f(x) - \tilde{f}(x) &= \sum_{j=1}^p u_j(x) (f(x) - f(v_j)) \\
    &=  \sum_{j,\ d(v_j,x) < \rho} u_j(x) (f(x) - f(v_j))
\end{split}
\end{equation}

This gives, using triangular inequality and $u_j(x)^2\leq u_j(x)$ (since $0\leq u_j(x)\leq 1$):

\begin{equation}
\begin{split}
    |f(x) - \tilde{f}(x)|^2 &\leq \sum_{j,\ d(v_j,x) < \rho} u_j(x)^2 |f(x) - f(v_j)|^2\\
    &\leq \sum_{j,\ d(v_j,x) < \rho} u_j(x) |f(x) - f(v_j)|^2\\
\end{split}
\end{equation}

Which gives $ |f(x) - \tilde{f}(x)|^2\leq \epsilon$ using the hypothesis of the proposition and the fact $\sum_j u_j = \mathbb{1}$.

Furthermore, if there exit $k$ so that $x=v_k$, we can remove the term of index $k$ and we have
\begin{equation}
\begin{split}
    |f(x) - \tilde{f}(x)|^2 &\leq \sum_{j\neq k,\ d(v_j,x) < \rho} u_j(v_k)^2 |f(v_k) - f(v_j)|^2\\
    &\leq \eps^2 \sum_{j\neq k,\ d(v_j,x) < \rho} u_j(v_k)\\
    &\leq \eps^2 (1 - u_k(v_k))
\end{split}
\end{equation}

\end{proof}

\begin{proof} \textit{of Prop.~\ref{prop:Approximation Error}}

We again suppose all shapes to be area-normalized. 
Using the $\tilde{f}$ notation from proposition~\ref{prop:interpolation error}, we can use the triangular inequality on $\left\|\bm{\Pi}\overline{\bm{\Psi}}^\Mm - \*U^\Nn\overline{\bm{\Pi}}\ \overline{\bm{\Phi}}^\Mm \right\|_\Nn$:

\begin{equation}
\begin{split}
    \left\|\bm{\Pi}\overline{\bm{\Psi}}^\Mm - \*U^\Nn\overline{\bm{\Pi}}\ \overline{\bm{\Phi}}^\Mm \right\|_\Nn^2 \leq &\left\|\bm{\Pi}\overline{\bm{\Psi}}^\Mm - \widetilde{\bm{\Pi}\overline{\bm{\Psi}}^\Mm} \right\|_\Nn^2 \\
    + &\left\|\widetilde{\bm{\Pi}\overline{\bm{\Psi}}^\Mm} - \*U^\Nn\overline{\bm{\Pi}}\ \overline{\bm{\Phi}}^\Mm \right\|_\Nn^2
    \end{split}
\end{equation}

The first term can be decomposed as a sum of the norms of its $K$ columns, where each term is in the form $\|\overline{\Psi}^\Mm_j\circ T - \widetilde{\overline{\Psi}^\Mm_j\circ T}\|^2_\Nn$, and can be controlled by applying the bound on interpolation error from proposition~\ref{prop:interpolation error} associated with the bounded distortion lemma, that is 
\begin{equation}\label{eq:proof:bound1}
    \left\|\bm{\Pi}\overline{\bm{\Psi}}^\Mm - \widetilde{\bm{\Pi}\overline{\bm{\Psi}}^\Mm} \right\|_\Nn^2 \leq K B_T^2\epsilon^2
\end{equation}

Focusing on the second term, the following lemma will be very useful in order to bound it :
\begin{lemma}\label{lemma:U norm}
Given $\beta \in\RR^{p^\Nn}$, $\|\*U^\Nn\beta\|_\Nn^2 \leq \|\beta\|_F^2$
\end{lemma} 

We indeed notice the second term can be written in the form $\|\*U^\Nn \*A - \*U^\Nn \*B\|_\Nn^2$. Using lemma~\ref{lemma:U norm}, we can now focus on bounding $\|A-B\|_F^2$ and especially on the squared norm of each of the columns of $A-B$. In practice, each column can be written as $\inpar{\overline{\psi}_j^\Mm\circ T(v_k) - \overline{\phi}_j^\Mm\circ T_{\vert\Ss^\Nn}(v_k)}_k$, since we supposed that $T_{\vert\Ss^\Nn}$ was well defined between the subsamples.
\\

Given $k\in\{1,\dots, p^\Nn\}$, there exists $i_0\in\{1,\dots,p^\Mm\}$ so that $T(v_i^\Nn) = v_{i_0}^\Mm$.\\
Furthermore, by definition of the approximated eigenvectors $\overline{\psi}^\Mm_j$, for all $x\in\Mm$ we have $\overline{\psi}^\Mm_j(x) = \sum_{k=1}^{p^\Mm} \overline{\phi}^\Mm_j(v_k^\Mm) u_k^\Mm(x)$\\
Therefore, denoting $\Delta_j(i) = \overline{\psi}^\Mm_j(v_{i_0}^\Mm) -\overline{\phi}_j^\Mm(v_{i_0}^\Mm)$

\begin{align}
    \Delta_j(i) &=
    \sum_{k=1}^{p^\Mm} \overline{\phi}^\Mm_j(v_k^\Mm) u_k^\Mm(v_{i_0}^\Mm) - \overline{\phi}^\Mm_j(v_{i_0}^\Mm) \\
    &= \sum_{k=1}^{p^\Mm}u_k^\Mm(v_{i_0}^\Mm)\inpar{\overline{\phi}^\Mm_j(v_k^\Mm) - \overline{\phi}^\Mm_j(v_{i_0}^\Mm)}
\end{align}

The exact same procedure as in the proof of Proposition~\ref{prop:interpolation error} can now be applied which allows to bound the term
\begin{equation}\label{eq:proof:bound2}
    \left\|\widetilde{\bm{\Pi}\overline{\bm{\Psi}}^\Mm} - \*U^\Nn\overline{\bm{\Pi}}\ \overline{\bm{\Phi}}^\Mm \right\|_\Nn^2 \leq K\eps (1-\al)
\end{equation}

Summing terms from~\eqref{eq:proof:bound1} and~\eqref{eq:proof:bound2} produce the upper bound of proposition~\ref{prop:Approximation Error}.



\end{proof}

\section{Values of theoretical quantities}
\label{app:parameters values}
We here provide values for the named values from Proposition~\ref{prop:Approximation Error}. We again highlight the proposed bounds are not tight and only serves as guidance to select parameters.

First, note that $B_T$ is a Lipchitz-constant, which is $1$ whenever $T$ is an isometry, and is else related to the area-distortion induced by $T$.

$\alpha$ can then vary between $0$ and $1$, but our adaptive radius scheme ensures the minimal value is $0.3$. In practice, the average value is higher, around $.43$ on average on the SHREC19 dataset.

Finally $\varepsilon$ controls the variation of the approximated eigenvectors in a local neighborhood, and can be set arbitrarily small by decreasing the value for $\rho$ (and potentially increasing the number of samples to ensure partition of unity). Note that the higher the frequency of the eigenvector, the higher the maximal value of $\varepsilon$ is, where the maximum is taken across all local neighborhoods.
In practice, we observe maximum values between $0$ and $8$ for the first $150$ eigenvectors, when using around $1500$ sampled points. In comparison, we obtain values between $0$ and $4$ by comparing values of the \emph{exact} eigenvectors simply across edges.

\section{Functional Map approximation}
\label{app:FM approx}
We here display on Figure~\ref{fig:app:FM approx} images of the ground truth functional maps for Figure~\ref{fig:FM vs reduced FM}

\begin{figure}
    \centering
    \includegraphics[width=.2\textwidth]{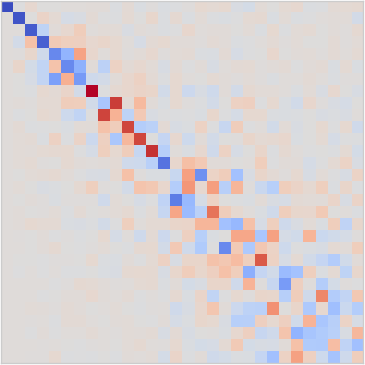}\hspace{.2cm}
    \includegraphics[width=.2\textwidth]{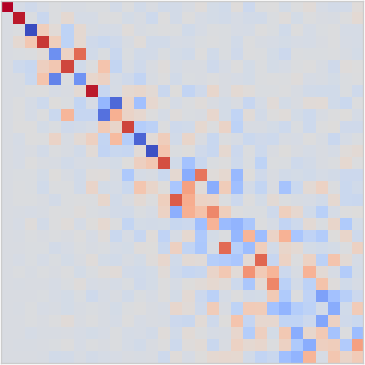}
    \caption{Ground truth functional map $\overline{\*C}$ using the functional space $\overline{\Ff}$ without (Left) and with (Right) adaptive radius. Notice that up to a change of sign, both functional maps look similar}
    \label{fig:app:FM approx}
\end{figure}

\section{Implementation details}
\label{app:implementation}
We here provide additional details on parameters and algorithm for implementation.

Per vertex radii are initially set to the same initial value $\rho_0$, defined as $\rho_0 =3\tilde{\rho}_0$ with $\tilde{\rho}_0=\sqrt{\frac{\text{Area}(\Mm)}{p\pi}}$. The value of $\tilde{\rho}_0$ is obtained by expecting each sample point $v_j$ to occupy a geodesic disk or radius $\tilde{\rho}_0$, which would eventually cover the complete shape - that is $p\pi\rho^2=\text{Area}(\Mm)$. If the choice of the sample is free, we recommend using Poisson Disk Sampling to obtain roughly evenly spaced samples in a fast manner. 

Local Dijkstra starting from samples can be accelerated by both parallelization and reduction of the search space to a euclidean ball of radius $\rho_0$ around each sample as we have $d^\Mm(x_i,x_j)\leq\|x_i-x_j\|_2$. 

If some points $x_i\in\Mm$ have not been reached during this process, one should either increase the initial radius $\rho_0$ or simply add $x_i$ to the sample set $\Ss$ and run an extra local Dijkstra starting from $x_i$.

Values of $\inpar{\tilde{u}_j}_j$ can now be computed and stored in a sparse $n\times p$ matrix $\widetilde{\*U}$ where each column stores a local function. Eventually in order to detect too small self-weights $u_j(v_j)$, we notice from Equation~\eqref{eq:self weight} that $u_j(v_j) \leq \al$ is equivalent to $\sum_{i} \tilde{u}_i(v_j) \geq \frac{1}{\al} $, where the first term is the sum of a row of a $p\times p$ submatrix extracted from $\widetilde{\*U}$. Reducing the radius $\rho_j$ of a sample only consists in recomputing the $j$-th column of $\widetilde{U}$ from the \emph{same} distance values as computed by the first Dijkstra run. This way, no additional Dijkstra is run which leads to a somewhat costless improvement of the local functions.

\section{Texture transfer}
\label{app:texture transfer}
We further show the efficiency of our method in terms of accuracy by displaying another example of texture transfer on a pair of dense meshes part of the SHREC19 dataset, as seen on Figure~\ref{fig:texture shrec2}. Here the leftmost shape contains $50\ 000$ vertices and rightmost $200\ 000$, but we only use nearly $1500$ samples to obtain such correpondences.
Note that in this case, where the number of vertices on the target shape is larger than the number of vertices on the source shape, texture transfer is especially challenging as multiple vertices of the target shape are projected into the same triangles. This makes texture transfer very sensitive to the quality of the estimated map.
We stress this pipeline obtains sub-sample accuracy in the correspondences, all in a fraction of the required time to run the exact ZoomOut pipeline.
We further highlight that texture at the elbows and shoulder is not smooth on the source shape, which explains the distortion on the target shape. This results simply serves as visualization.

\begin{figure}[htb]
  \centering
  \includegraphics[width=.9\linewidth]{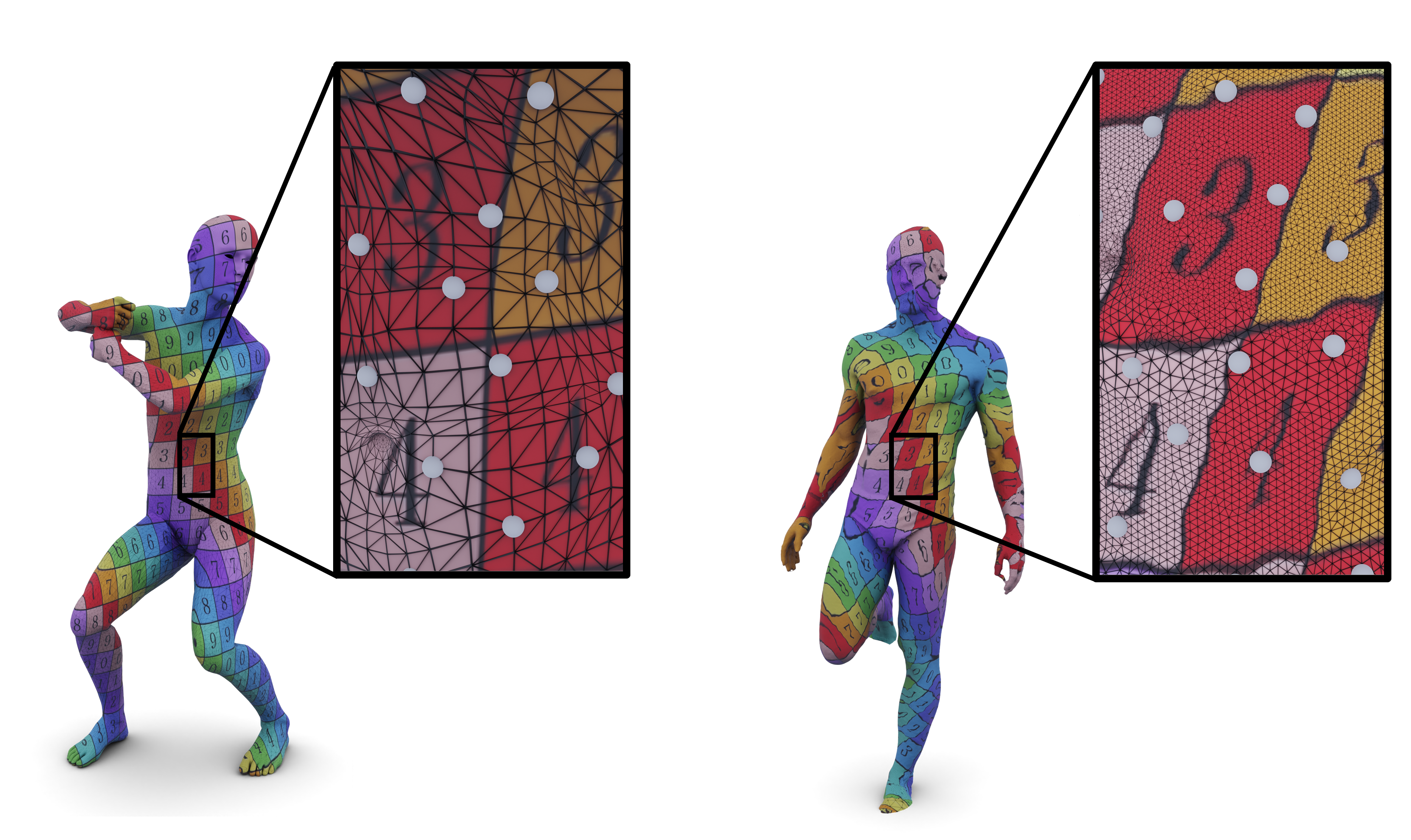}
  \caption{\label{fig:texture shrec2}
           Texture transfer using our scalable version of ZoomOut on a pair of the SHREC19 dataset. Samples used in the pipeline are shown as white dots.}
\end{figure}

\end{document}